\documentclass{article}

\usepackage{paralist}
\usepackage{amsmath}
\usepackage{amsthm}
\usepackage{amssymb}
\usepackage{amsfonts}
\usepackage{array}
\usepackage{paralist}
\usepackage{fullpage}
\usepackage{tikz}

\newtheorem{theorem}{Theorem}[section]

\newtheorem{fact}[theorem]{Fact}
\newtheorem{lemma}[theorem]{Lemma}

\newtheorem{remark}[theorem]{Remark}

\newtheorem{definition}[theorem]{Definition}
\newtheorem{proposition}[theorem]{Proposition}

\newcommand{\RR}{\mathbb R}

\newcommand{\NN}{\mathbb N}

\newcommand{\E}{\mbox{\bf E}}

\newcommand{\cM}{{\cal M}}

\newcommand{\cQ}{{\cal Q}}
\newcommand{\cR}{{\cal R}}
\newcommand{\cX}{{\cal X}}

\def\b1{{\bf 1}}

\def\bx{{\bf x}}

\newcommand{\set}[1]{\left\{ #1 \right\}}

\renewcommand{\hat}{\widehat}
\renewcommand{\tilde}{\widetilde}

\newcommand{\RRp}{\RR_+}

\def\A{\mathcal{A}}
\def\B{\mathcal{B}}

\def\U{\mathcal{U}}
\def\V{\mathcal{V}}
\def\sse{\subseteq}
\def\poly{\mbox{poly}}

\title{Limitations of Randomized Mechanisms for Combinatorial Auctions}

\author{
Shaddin Dughmi\thanks{Department of Computer Science, Stanford University, 460 Gates Building, 353 Serra Mall, Stanford, CA 94305.  Supported by NSF Grant CCF-0448664. Email: {\tt shaddin@cs.stanford.edu}.} 
\and 
Jan Vondr\'ak\thanks{IBM Almaden Research Center, 650 Harry Rd, San Jose, CA 95120. E-mail: {\tt jvondrak@us.ibm.com}.}
}
\begin{document}
\maketitle
\begin{abstract}

The design of computationally efficient and incentive compatible mechanisms  that solve or approximate fundamental resource allocation problems is the main goal of algorithmic mechanism design.
A central example in both theory and practice is welfare-maximization in combinatorial auctions.
Recently, a randomized mechanism has been discovered for combinatorial auctions
that is truthful in expectation and guarantees a $(1-1/e)$-approximation to the optimal social welfare when players have coverage valuations
\cite{DRY11}. This approximation ratio is the best possible even for non-truthful algorithms, assuming $P \neq NP$ \cite{KLMM05}.

Given the recent sequence of negative results for combinatorial auctions under more restrictive notions of incentive compatibility \cite{DN07,BDFKMPSSU10,Dobzin11}, this development raises a natural question:
Are truthful-in-expectation mechanisms compatible with polynomial-time approximation in a way
that deterministic or universally truthful mechanisms are not? In particular, can polynomial-time truthful-in-expectation mechanisms guarantee a near-optimal approximation ratio for more general variants of combinatorial auctions?

We prove that this is not the case. Specifically, the result of \cite{DRY11} cannot be extended to combinatorial auctions
with submodular valuations in the value oracle model. (Absent strategic considerations, a $(1-1/e)$-approximation is still achievable
in this setting \cite{V08}.) More precisely, we prove that there is a constant $\gamma>0$
such that there is no randomized mechanism that is truthful-in-expectation--- or even approximately truthful-in-expectation ---
and guarantees an $m^{-\gamma}$-approximation to the optimal social welfare for combinatorial auctions with submodular
valuations in the value oracle model.

We also prove an analogous result for the flexible combinatorial public projects (CPP) problem, where
a truthful-in-expectation $(1-1/e)$-approximation for coverage valuations has been recently developed \cite{Dughmi11}.
We show that there is no truthful-in-expectation --- or even approximately truthful-in-expectation --- mechanism that achieves an $m^{-\gamma}$-approximation to the optimal social welfare for combinatorial public projects with submodular valuations
in the value oracle model. 
Both our results present an unexpected separation between coverage functions and submodular functions,
which does not occur for these problems without strategic considerations.
\end{abstract}

\thispagestyle{empty} 
\addtocounter{page}{-1} 
\newpage

\section{Introduction}
\label{sec:intro}

The design of incentive-compatible mechanisms for welfare maximization in combinatorial auctions is a central problem of algorithmic mechanism design. In a combinatorial auction, there are $n$ players and a set $M$ of $m$ items. Player $i$ has a (private) valuation function $v_i:2^M \rightarrow \RR_+$ which is assumed to be monotone ($v_i(S) \leq v_i(T)$ whenever $S \subset T$) and normalized ($v_i(\emptyset) = 0$). The goal is to design a computationally efficient mechanism that yields an allocation of items $(S_1,\ldots,S_n)$ to the players along with payments $(p_1,\ldots,p_n)$ so that (a) the {\em social welfare} $\sum_{i=1}^{n} v_i(S_i)$ is approximately maximized, and (b) the mechanism is {\em incentive-compatible}, or {\em truthful}, meaning that each player maximizes his \emph{utility} $v_i(S_i) - p_i$ by reporting his true valuation $v_i$.

This problem has been studied extensively in both strategic and non-strategic settings. Various strategic solution concepts have been considered, including deterministic truthfulness, universal truthfulness, and truthfulness in expectation. Moreover, both strategic and non-strategic formulations of the problem have been studied for various restricted classes of valuations, as well as under various assumptions on how valuations are accessed or represented. Absent assumptions on the class of valuations, the welfare maximization problem is very hard to approximate even by non-truthful algorithms (NP-hardness of  $m^{\epsilon-1/2}$-approximation follows from the set packing problem). Better approximation ratios are possible for valuation classes that restrict complementarity between items. The most prominent such class of valuations is \emph{submodular functions}: functions $v_i$ where the marginal value $v_i(S \cup \{j\}) - v_i(S)$ for a each fixed item $j$ is non-increasing in $S$. It is known that the welfare maximization problem with submodular valuation functions admits a (non-truthful) $(1-1/e)$-approximation algorithm \cite{V08}, and this is optimal  assuming $P \neq NP$ \cite{KLMM05}. The hardness result of \cite{KLMM05} holds even in the special case of coverage valuations; the algorithmic result of \cite{V08} holds in the {\em value oracle model}, where each $v_i$ can be queried only through an oracle returning $v_i(S)$ for a given query $S$. 
In the value oracle model, it is known that any $(1-1/e+\epsilon)$-approximation for combinatorial auctions with submodular valuations would require an exponential number of queries \cite{MSV08}. This is also the model we consider in this paper.

The classical VCG mechanism is incentive compatible and maximizes welfare in combinatorial auctions. Unfortunately, however, VCG can not be implemented in polynomial time even for very special classes of valuation functions, including submodular functions. Combining computational efficiency and truthfulness for combinatorial auctions appears difficult. A series of works have provided evidence that computational efficiency and truthfulness are in conflict: (deterministic) VCG-type mechanisms have been ruled out for submodular combinatorial auctions in the communication complexity model \cite{DN07}, and even for explicitly given budget-additive valuations \cite{BDFKMPSSU10}. Recently,  Dobzinski \cite{Dobzin11} proved that there is no deterministic truthful or even randomized universally truthful mechanism for submodular combinatorial auctions in the value oracle model, achieving an approximation ratio better than $m^{\epsilon-1/2}$.

Therefore, it came as a surprise when a $(1-1/e)$-approximate randomized mechanism was discovered by Dughmi, Roughgarden and Yan \cite{DRY11} for a large subclass of submodular valuations. Their mechanism is {\em truthful in expectation} --- a weaker notion than truthfulness in the universal sense --- and applies to explicitly represented coverage functions. More generally, their mechanism applies to ``black-box'' valuations that are expressible as weighted sums of matroid rank functions, provided they support ``lottery-value queries" (what is the expected value $\E[v_i(\hat{\bx})]$ for a given product distribution $\hat{\bx}$). 
The mechanism can be also implemented in the value oracle model, 
at the cost of relaxing the solution concept to approximate truthfulness in expectation \cite{DRVY11}.

This  development  raises a natural question: Could truthfulness-in-expectation be the cure for combinatorial auctions, perhaps providing an optimal $(1-1/e)$-approximation for all submodular valuations? Given that a $(1-1/e)$-approximation for welfare maximization in combinatorial auctions (without truthfulness) was also discovered first for coverage functions \cite{DS06}, then for weighted sums of matroid rank functions \cite{CCPV07} and later extended to monotone submodular functions \cite{V08}, it seems reasonable to conjecture that the same might happen for truthful-in-expectation mechanisms.

\paragraph{Our results.}
We prove that this is not the case,  and there is a significant separation between the class of coverage functions and general monotone submodular functions. More precisely, there is no truthful-in-expectation mechanism (even $(1-\epsilon)$-approximately truthful-in-expectation) for submodular combinatorial auctions in the value oracle model, guaranteeing an approximation better than $1/m^\gamma$ for some fixed $\epsilon,\gamma>0$ (Theorem~\ref{thm:CA-hardness}). In particular, the results of \cite{DRY11} cannot be extended to all monotone submodular functions.

We also prove a similar result for the \emph{flexible submodular combinatorial public projects} problem (see Section \ref{sec:CPP-hardness} for a history of this problem): there is no $(1-\epsilon)$-approximately truthful-in-expectation mechanism providing approximation better than $1/m^\gamma$ for some $\gamma>0$. This is true even in the case of a single player. The combinatorial public projects problem admits a simpler structure than combinatorial auctions,  and hence we use it as a warm-up to demonstrate our approach.

\begin{figure}[here]
$\begin{array}{|| c || c | c | c ||} \hline
\mbox{Class of valuations} & \mbox{Approximation} & \mbox{Universally truthful} & \mbox{Truthful-in-expectation}  \\
\hline \hline
\mbox{submodular / value oracle} & 1-1/e & m^{-1/2} \mid m^{\epsilon-1/2} & m^{-1/2} \mid {m^{-\gamma}} \mbox{\bf [new]} \\
\hline
\mbox{coverage, matroid rank sums} & 1-1/e & m^{-1/2} \mid 1-1/e & 1-1/e \\
\hline
\mbox{budget-additive} & \frac{3}{4} \mid \frac{15}{16} & m^{-1/2} \mid \frac{15}{16} & m^{-1/2} \mid \frac{15}{16} \\
\hline
\mbox{submodular / demand oracle} & 1-1/e+\epsilon \mid \frac{15}{16} & \Omega(1/\log m \log \log m) \mid \frac{15}{16} & \Omega(1/\log m \log \log m) \mid \frac{15}{16} \\
\hline
\end{array} $
\caption{\small Currently known results for combinatorial auctions: approximation $|$ inapproximability.
If only one result is given, it is known to be optimal. For randomized maximal-in-range (universally truthful) mechanisms, it is known that it is hard to achieve a better than $1/n$-approximation for coverage valuations; however, other universally truthful mechanisms might exist. No non-trivial hardness was previously known for truthful-in-expectation combinatorial auctions, even when restricted to maximal-in-distributional-range mechanisms.}
\end{figure}

\paragraph{Our techniques.}
Our hardness results are obtained by combining two recently developed techniques: the {\em symmetry gap} technique for submodular functions \cite{V09}, and the {\em direct hardness} approach for combinatorial auctions \cite{Dobzin11}.

First, we consider the possibility of maximal-in-distributional range (MIDR) mechanisms.
We endeavor to explain why the approach of \cite{DRY11} breaks down when applied to monotone submodular functions. The answer lies in a certain convexity phenomenon that can be exploited in a symmetry gap argument. The symmetry gap argument on its own rules out the approach of \cite{DRY11}. Furthermore, it is possible to generalize the argument to an arbitrary MIDR mechanism, and moreover amplify the gap to some constant power of $m$. In fact our approach rules out even non-uniform approximately-MIDR mechanisms.

In the case of combinatorial public projects (CPP), we prove that if non-uniformity is allowed, then  approximately truthful-in-expectation mechanisms are no more powerful --- in terms of approximating combinatorial auctions using a polynomial number of value queries --- than  MIDR mechanisms. Therefore, by ruling out MIDR mechanisms, we also rule out  truthful-in-expectation mechanisms.
In the case of combinatorial auctions, no such  equivalence in power between truthful-in-expectation and MIDR mechanisms is known. Instead, we apply the direct hardness approach of Dobzinski \cite{Dobzin11} to identify a single player for whom the allocation problem in some sense mimics the CPP problem. Again, the symmetry gap argument can be used here, though payments complicate  the picture. We address this difficulty by employing a scaling argument and invoking the separating hyperplane theorem --- this allows us to essentially get rid of the payments and use the same gap amplification technique we used for the CPP problem to obtain a hardness of $m^{-\gamma}$-approximation.

\paragraph{Organization of the paper.}
After the necessary preliminaries (Section~\ref{sec:prelims}), we present our intuition on the separation between coverage and submodular functions in Section~\ref{sec:intuition}. In Section~\ref{sec:CPP-hardness}, we present an overview of the proof of hardness for combinatorial public projects, and in Section~\ref{sec:auctions-hardness} an overview of the proof for combinatorial auctions. The complete proofs are deferred to the appendices.


\section{Preliminaries}
\label{sec:prelims}

\subsection{Mechanism Design Basics}\label{sec:MD}

\paragraph{Mechanism Design Problems.} We consider mechanism design problems where there are $n$ players, and a set $\Omega$ of feasible solutions. Each player $i$ has a non-negative \emph{valuation function} $v_i: \Omega \to \RRp$. We are concerned with \emph{welfare maximization} problems, where the objective is $\sum_{i=1}^n v_i(\omega)$.

\paragraph{Mechanisms.} We consider direct-revelation mechanisms for mechanism design problems.  Such a mechanism  comprises an {\em
  allocation rule} $\A$, which is a function from (hopefully
truthfully) reported valuation functions $v=(v_1,\ldots,v_n)$ to an outcome
$\A(v) \in \Omega$, and a {\em payment rule} $p$, which is a function from
reported valuation functions to a required payment $p_i(v)$ from each player $i$.
We allow the allocation and payment rules to be randomized. We restrict our attention to mechanisms that are individually rational in expectation --- i.e. $\E[v_i(\A(v)) - p_i(v)] \geq 0 $ --- and  the payments are non-negative in expectation --- i.e. $\E[p_i(v)] \geq 0$ --- for each player $i$ and each input $v=(v_1,\ldots,v_n)$, when the expectations are over the random coins of the mechanism. 

\paragraph{Truthfulness.} A mechanism with allocation and payment rules $\A$ and $p$ is {\em
  truthful-in-expectation} if every player always maximizes its expected
  payoff by truthfully reporting its valuation function, meaning that
\begin{equation}\label{eq:truthful}
\E[v_i(\A(v)) - p_i(v)] \geq \E[ v_i(\A(v'_i,v_{-i})) -
  p_i(v'_i,v_{-i})]
\end{equation}
for every player~$i$, (true) valuation function~$v_i$, (reported)
valuation function~$v'_i$, and (reported) valuation functions~$v_{-i}$ of
the other players.
The expectation in~\eqref{eq:truthful} is over the coin flips of the
mechanism.  If \eqref{eq:truthful} holds for every flip of the coins, rather than merely in expectation, we call the mechanism \emph{universally truthful}.

\paragraph{VCG-Based Mechanisms.} Mechanisms for welfare maximization problems are often variants of the classical VCG mechanism.  
Recall that the {\em VCG
  mechanism} is defined by the (generally intractable)
allocation rule that selects the welfare-maximizing outcome with
respect to the reported valuation functions, and the payment rule that
charges each player~$i$ a bid-independent ``pivot term'' minus the
reported welfare earned by other players in the selected outcome.  This
(deterministic) mechanism is truthful; see e.g.~\cite{Nis07}.

Let $dist(\Omega)$ denote the probability distributions over the set of 
feasible solutions $\Omega$, and let $\cR \sse dist(\Omega)$ be a compact subset of
them.  The corresponding {\em Maximal in Distributional Range (MIDR)}
allocation rule is defined as follows: given reported valuation
functions $v_1,\ldots,v_n$, return an outcome that is sampled randomly
from a distribution $D^* \in \cR$ that maximizes the expected welfare
$\E_{\omega \sim D}[\sum_i v_i(\omega)]$ over all distributions $D \in \cR$.
Analogous to the VCG mechanism, there is a (randomized) payment rule
that can be coupled with this allocation rule to yield a
truthful-in-expectation mechanism (see~\cite{DD09}). We note that deterministic MIDR allocation rules
 --- i.e. those where $\cR$ is a set of point distributions --- are called \emph{maximal-in-range (MIR)}.

\paragraph{Approximate Truthfulness.}
For $\epsilon\geq 0$, a  mechanism with allocation and payment rules $\A$ and $p$ is {\em $(1-\epsilon)$-approximately   truthful-in-expectation} if 
\begin{equation}\label{eq:approx-truthful}
\E[v_i(\A(v)) - p_i(v)] \geq (1-\epsilon) \E[ v_i(\A(v'_i,v_{-i})) -
  p_i(v'_i,v_{-i})]
\end{equation}
for every player~$i$, (true) valuation function~$v_i$, (reported)
valuation function~$v'_i$, and (reported) valuation functions~$v_{-i}$ of
the other players.
The expectation in~\eqref{eq:approx-truthful} is over the coin flips of the
mechanism.  
Using the fact that payments are non-negative in expectation, a $(1-\epsilon)$-approximately truthful-in-expectation mechanism also satisfies the following weaker condition. (This condition is sufficient for our hardness results.)
\begin{equation}\label{eq:approx-truthful-relaxed}
\E[v_i(\A(v)) - p_i(v)] \geq  \E[ (1-\epsilon) v_i(\A(v'_i,v_{-i})) - p_i(v'_i,v_{-i})]
\end{equation}
Approximately truthful mechanisms are related to \emph{approximately maximal-in-distributional-range} allocation rules. An allocation rule $\A: \V \to \Omega$ is $(1-\epsilon)$-approximately maximal-in-distributional range if it fixes a $\cR \sse dist(\Omega)$, and returns an outcome that is sampled from $D^* \in \cR$ that $(1-\epsilon)$-approximately maximizes the expected welfare
$\E_{\omega \sim D}[\sum_i v_i(\omega)]$ over all distributions $D \in \cR$.  We show in Appendix  \ref{sec:TIE-MIDR} a sense in which approximately maximal-in-distributional-range allocation rules are no less powerful -- in terms of approximating the social welfare -- than approximately truthful-in-expectation mechanisms.

Our main reason for considering the notion of approximate truthfulness is that the mechanisms of \cite{DRY11,Dughmi11}, if implemented in the value oracle model, are only approximately truthful-in-expectation (for an arbitrarily small $\epsilon>0$) \cite{DRVY11}. The value oracle model seems too weak to make the mechanisms of \cite{DRY11,Dughmi11} exactly truthful-in-expectation; however, \cite{DRVY11}  makes it quite conceivable that there might be an approximately truthful-in-expectation mechanism for combinatorial auctions and combinatorial public projects, both with submodular valuations.

\subsection{Combinatorial Auctions}

In \emph{Combinatorial Auctions} there is a set $M$ of $m$ items, and a set of $n$ players. Each player $i$ has a valuation
function $v_i:2^M\rightarrow \RRp$ that is normalized
($v_i(\emptyset) = 0$) and monotone ($v_i(A) \leq v_i(B)$ whenever
$A \sse B$). A feasible solution is an \emph{allocation}
$(S_1,\ldots,S_n)$, where $S_i$ denotes the items assigned to player $i$,
and $\set{S_i}_i$ are mutually disjoint subsets of $M$.
Player $i$'s value for outcome $(S_1,\ldots,S_n)$ is equal to $v_i(S_i)$.  The
goal is to choose an allocation maximizing \emph{social welfare}: $\sum_i v_i(S_i)$.

\subsection{Combinatorial Public Projects}

In \emph{Combinatorial Public Projects} there is a set
$[m]= \set{1,\ldots,m}$ of \emph{projects}, a cardinality bound $k$ such that $0 \leq k \leq m$, and a set
$[n]=\set{1,\ldots,n}$ of \emph{players}. Each player $i$ has a valuation
function $v_i:2^{[m]}\rightarrow \RRp$ that is normalized
($v_i(\emptyset) = 0$) and monotone ($v_i(A) \leq v_i(B)$ whenever
$A \sse B$). In this paper, we focus on the \emph{flexible} variant of combinatorial public projects: a feasible solution is a set $S \sse [m]$ of projects with $|S| \leq k$. Player $i$'s value for outcome $S$ is equal to $v_i(S)$. Prior work \cite{PSS08,BSS10,Dobzin11} has also considered the \emph{exact} variant, where a feasible solution is a set $S \sse [m]$ with $|S| = k$. In both variants, the
goal is to choose a feasible set $S$ maximizing \emph{social welfare}: $\sum_i v_i(S)$.


\section{Intuition - what fails for submodular valuations}
\label{sec:intuition}
\label{sec:exp-gap}

The main obstacle in proving our hardness result for submodular functions is the fact that the natural subclass of coverage functions {\em does} admit a truthful-in-expectation $(1-1/e)$-approximation \cite{DRY11}. In the absence of strategic considerations, coverage functions capture the full difficulty of submodular functions in the context of welfare maximization, in the sense that they exhibit the same hardness threshold of $1-1/e$. Hence, it is not immediately clear where the dramatic jump in hardness should come from.


Let us recall the main idea of \cite{DRY11}: Let $f:2^M \to \RRp$ be a submodular set function. Given $\bx \in [0,1]^M$, the expected value of $f(S)$ when $S$ includes each item $j$ independently with probability $x_j$ is measured by the {\em multilinear extension} $F(\bx)$, which has been previously used in work on submodular maximization \cite{CCPV07,V08,KST09,V09,OV11}. $F$ is an \emph{extension} of $f$, in the sense that it agrees with $f$ on integer points, and therefore maximizing $F(\bx)$ over fractional allocations would yield an optimal algorithm. However, $F(\bx)$ is not a concave function and can be maximized only approximately. 
Instead, the authors of \cite{DRY11} consider a different rounding process --- which they call the \emph{Poisson rounding scheme} --- that includes each $j$ in $S$ with probability $1-e^{-x_j}$ instead. The expected value of applying the Poisson rounding rounding scheme to a point $\bx$ is measured by a modified function $F^{exp}(x_1,\ldots,x_m) = F(1-e^{-x_1},\ldots,1-e^{-x_m})$, which fortuitously {\em turns out to be concave} for a subclass of submodular functions, including coverage functions and weighted sums of matroid rank functions. In this case,  $F^{exp}(\bx)$ can be maximized exactly, and yields a maximal-in-distributional-range algorithm whose range is the image of the Poisson rounding scheme.
Since the ratio between $F(\bx)$ and $F^{exp}(\bx)$ is bounded by $1-1/e$, this leads to a truthful-in-expectation $(1-1/e)$-approximation.

The first question is whether $F^{exp}$ can be maximized for any monotone submodular function. 
It was observed by the authors of \cite{DRY11} that $F^{exp}$ is not concave for every submodular function:
one example is the budget-additive function $f(S) = \min \{ \sum_{i \in S} w_i, 2 \}$ where $w_1=w_2=w_3=1$ and $w_4=2$. Hence convex optimization techniques cannot be used for $F^{exp}(\bx)$ directly; still, perhaps $F^{exp}(\bx)$ could be maximized for a different reason.
We prove that this is impossible, using a {\em symmetry gap} argument \cite{FMV07,MSV08,V09}. 

The budget-additive function above does not lend itself well to the symmetry gap argument, because there is a clear asymmetry between the elements of weight 1 and the element of weight 2. Instead, we construct an example where $F^{exp}$ is not concave and all elements are in some sense ``equivalent". For this purpose, we use the following construction: If $f_1,f_2:2^M \rightarrow [0,1]$ are monotone submodular functions, then 
$$f(S) = 1 - (1-f_1(S))(1-f_2(S))$$
is also a monotone submodular function (see Lemma~\ref{lem:submod-product}).
In particular, let $M = M_1 \cup M_2$, $|M_1| = |M_2| = m$,
$|M|=2m$, and let $f_i(S) = \min \{ \frac{1}{\alpha m} |S \cap M_i|, 1 \}$ for some $\alpha > 0$. These are budget-additive and hence monotone submodular functions. Then we set 
$$ f(S) = 1 - (1-f_1(S))(1-f_2(S)) = 1 - \left(1 - \frac{1}{\alpha m} |S \cap M_1|\right)_+ \left(1 - \frac{1}{\alpha m} |S \cap M_2| \right)_+.$$
Here, $(y)_+ = \max \{ y, 0\}$ denotes the positive part of a number. By Lemma~\ref{lem:submod-product}, $f(S)$ is a monotone submodular function.
Let's consider the function $F^{exp}(x_1,\ldots,x_{2m}) = F(1-e^{-x_1},\ldots,1-e^{-x_{2m}})$. If $m \rightarrow \infty$, a random set obtained by sampling with probabilities $1-e^{-x_i}$ will have cardinality very close to $\sum (1-e^{-x_i})$.
We obtain
$$ F^{exp}(\bx) \simeq 1 - \left(1 - \frac{1}{\alpha m} \sum_{i \in M_1} (1-e^{-x_i}) \right)_+
 \left(1 - \frac{1}{\alpha m} \sum_{j \in M_2} (1-e^{-x_j}) \right)_+.$$
The reader can verify that this function is concave for $\alpha=1$.
But this is a very special coincidence. (The reason is that $f$ for $\alpha=1$ can be represented as a coverage function.)
Any smaller value of $\alpha$, for instance $\alpha = 1/2$, gives a non-concave function $F^{exp}$,
as can be seen by checking $\bx = \b1_{M_1}$, $\bx = \b1_{M_2}$ and $\bx = \frac12 \b1_M$:
$ F^{exp}(\b1_{M_1}) = F^{exp}_2(\b1_{M_2}) = 1 - (-1 + 2 e^{-1})_+ = 1 $
(note that $-1 + 2e^{-1} < 0$), while the value at the midpoint is
$ F^{exp}\left(\frac12 \b1_M \right) \simeq 1 - (-1 + 2 e^{-1/2})^2 = 4 e^{-1/2} - 4 e^{-1} \simeq 0.955.$
Therefore, we have an example where $F^{exp}(\bx)$ is not concave and moreover, all elements play the same symmetric role in $f$.
(Formally, $f$ has an element-transitive group of symmetries.)
Functions of this type will play a crucial role in our proof.

\paragraph{The symmetry gap argument.}
The symmetry gap argument from \cite{V09}, building up on previous work \cite{FMV07, MSV08}, shows the following:
Instances exhibiting some kind of symmetry can be blown up and modified in such a way that the only solutions that an algorithm can find (using a polynomial number of value queries) are symmetric with respect to the same notion of symmetry.
Thus the gap between symmetric and asymmetric solutions implies an inapproximability threshold.
We use this argument here as follows. The instance above (for $\alpha=1/2$) can be slightly modified as in \cite{FMV07,MSV08,V09}, in such a way that it is impossible to find any solution that is asymmetric with respect to $M_1, M_2$. 
Consider the optimization problem
$$ \max \{F^{exp}(\bx): \sum x_i \leq m \}.$$
The best symmetric solution is $F^{exp}(\frac12 \b1_M) \simeq 0.955$, while the optimum is $F^{exp}(\b1_{M_1}) = 1$. The only solutions found by a polynomial number of value queries are the symmetric ones, and hence we cannot solve the optimization problem within a factor better than $0.955$. A similar argument shows that we cannot solve the welfare maximization problem (for 2 players) with respect to $F^{exp}(\bx)$ within a factor better than $0.955$.

In the following, we harness this construction towards showing that there can be no good maximum-in-distributional-range mechanism, and eventually, no good truthful-in-expectation mechanism.


\section{Hardness for combinatorial public projects}
\label{sec:CPP-hardness}

We start with the combinatorial public project problem.
The (exact) combinatorial public project problem was introduced in \cite{PSS08} as a model problem for the study of truthful approximation mechanisms. This problem is better understood than combinatorial auctions, in the sense that a useful characterization of all deterministic truthful mechanisms is known: every truthful mechanism for 2 players is an \emph{affine maximizer} --- a weighted generalization of maximal-in-range mechanisms \cite{PSS08}. 
Using this characterization, it was proved in \cite{PSS08} that the exact submodular CPP problem does not admit any (deterministic) truthful $m^{\epsilon-1/2}$-approximation using a subexponential amount of communication, and moreover there is no  $m^{\epsilon-1/2}$-approximation even for a certain class of succintly represented submodular valuations unless $NP \subseteq BPP$. In contrast, the simple greedy algorithm is a non-truthful $(1-1/e)$-approximation algorithm for this problem \cite{NWF78}. This was the first example of such a dramatic gap in approximability between truthful mechanisms and non-truthful algorithms.

In follow-up work, a simpler 
characterization-type statement for CPP was shown in \cite{BSS10}: Every truthful mechanism for a single player with a coverage valuation can, via a non-uniform polynomial time reduction, be converted to a truthful maximal-in-range mechanism without degrading its approximation ratio. Since every truthful mechanism for $n$ players must embed a truthful mechanism for a single player, this allowed the authors to restrict attention to maximal-in-range mechanisms for a single player in proving an $m^{\epsilon-1/2}$-approximation threshold for CPP with coverage valuations, assuming that $NP \not\sse P/poly$. The following easy converse of their characterization is notable: A maximal-in-range mechanism for CPP with a single player can directly be used as a maximal-in-range mechanism for any number of players.

Recently, it was proved by Dobzinski \cite{Dobzin11} that the exact variant of the submodular CPP problem
(under the constraint $|S|=k$) does not admit a truthful-in-expectation $m^{\epsilon-1/2}$-approximation in the value oracle model. However, as noted in \cite{Dobzin11},  the flexible variant of CPP (under the constraint $|S| \leq k$) 
is arguably  more natural in the strategic setting. For the flexible variant of CPP, \cite{Dobzin11} proves that there is no universally truthful $m^{\epsilon-1/2}$-approximation, but leaves open the possibility of a better truthful-in-expectation mechanism.  Problems that have a packing structure like flexible CPP have historically proven to be easier to approximate using truthful-in-expectation mechanisms \cite{LS05,DD09,DR10,DRY11}. Flexible CPP has exhibited a similar pattern; Dughmi \cite{Dughmi11} recently designed a truthful-in-expectation $(1-1/e)$-approximation mechanism for CPP when players have explicit coverage valuations (which is optimal regardless of strategic issues \cite{Feige98}), and more generally when players have matroid rank sum valuations that support a certain randomized variant of value queries. 

\paragraph{Transformation to MIDR mechanisms.}
While deterministic truthful mechanisms for the CPP problem are no more powerful in terms of approximation than maximal-in-range mechanisms \cite{PSS08,BSS10}, the situation is slightly more complicated for randomized  mechanisms. It is not clear whether truthful-in-expectation mechanisms are equivalent to maximal-in-distributional-range mechanisms. Nonetheless, we prove the following.

\begin{theorem}
\label{thm:TIE=MIDR}
For every $\epsilon \geq 0$ and $c(m)>0$ the following holds. If there is a $(1-\epsilon)$-approximately truthful-in-expectation mechanism $\cM$ for the (exact or flexible) CPP problem that achieves a $c(m)$-approximation for submodular valuations on $m$ elements, then for any $\delta>0$ there is a non-uniform  $(1-3 \epsilon - \delta)$-approximately maximal-in-distributional-range mechanism $\cM'$ that achieves a $c(m)$-approximation for submodular valuations on $m$ elements and uses at most $m$ more value queries than $\cM$.
\end{theorem}


By a non-uniform mechanism, we mean a separate fixed mechanism for each input size $m$; i.e., the size of the program can depend arbitrarily on $m$. The only bound on the non-uniform mechanism is the number of value queries used.  The main idea is that the although the range of prices offered by a truthful-in-expectation mechanism can be unbounded, the mechanism can be made MIDR ``in the limit", when the input valuation is scaled by a sufficiently large constant. This constant can be fixed for each input size $m$ and acts as an ``advice string'' to the mechanism. We present the proof in Appendix~\ref{sec:TIE-MIDR}.

\paragraph{Hardness for MIDR mechanisms.}
Our hardness result for flexible submodular CPP rules out mechanisms purely based on the number of value queries used,
and hence it rules out even the non-uniform mechanisms mentioned in Theorem~\ref{thm:TIE=MIDR}.

\begin{theorem}
\label{thm:CPP-hardness}
There are absolute constants $\epsilon, \gamma > 0$ such that there is no $(1-\epsilon)$-approximately maximal-in-distributional-range mechanism for the flexible submodular CPP problem with 1 player in the value oracle model, $\max \{f(S): |S| \leq k\}$, achieving a better than $1/m^\gamma$-approximation in expectation in the objective function, where $m$ is the size of the ground set. This holds even for non-uniform mechanisms of arbitrary computational complexity, as long as the number of value queries is bounded by $\poly(m)$.
\end{theorem}

In the following, we present a sketch of the proof of this theorem. The full proof appears in Appendix~\ref{sec:CPP-proof}.

\paragraph{Proof strategy.}
We assume that a mechanism optimizes over a range of distributions $\cR$. (We assume for simplicity that the mechanism is MIDR rather than approximately MIDR.) We emphasize that the range $\cR$ is fixed beforehand, and the mechanism must optimize over $\cR$ for any particular submodular function $f$. This gives us a lot of flexibility in arguing about the properties of $\cR$. 


Suppose that the size of the ground set is $m = 2^{O(\ell)}$ and the cardinality bound is $k = m / 2^\ell$. We consider $\ell+1$ different ``levels" of valuation functions. (See Figure~\ref{fig:bisect-seq}.) At level $0$, we have a set $A^{(0)}$ of $m/2^\ell$ items, where the valuation function is nonzero and additive. Assuming that the mechanism achieves a $c$-approximation, there must be a distribution $D_0 \in \cR$ which allocates at least a $c$-fraction of $A^{(0)}$ in expectation to player $i$. This must be true for every set $A^{(0)}$ of size $m/2^\ell$. It will be useful to think of this set as random (and hidden from the mechanism.)

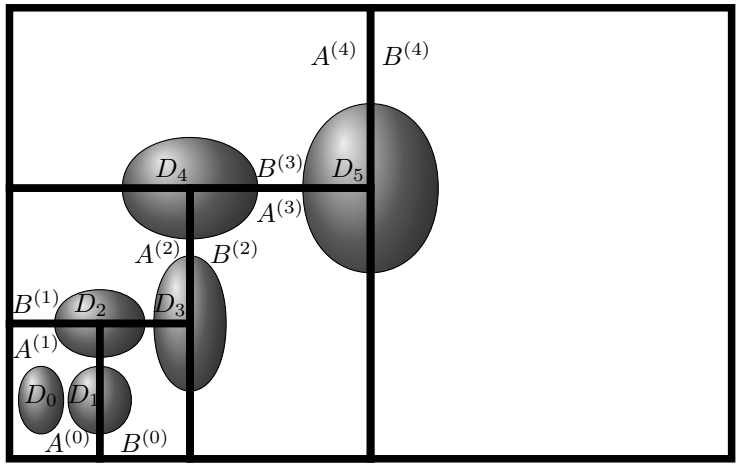
\begin{figure}[!ht]
\centering
\begin{tikzpicture}[scale=.60,pre/.style={<-,shorten <=2pt,>=stealth,thick}, post/.style={->,shorten >=1pt,>=stealth,thick}]

\shadedraw [black,ball color=gray] (0.2,1.3) .. controls +(0,1) and +(0,1) .. (1.2,1.3) .. controls +(0,-1) and +(0,-1) .. (0.2,1.3);
\shadedraw [black,ball color=gray] (1.3,1.3) .. controls +(0,1) and +(0,1) .. (2.7,1.3) .. controls +(0,-1) and +(0,-1) .. (1.3,1.3);
\shadedraw [black,ball color=gray] (1,3) .. controls +(0,1) and +(0,1) .. (3,3) .. controls +(0,-1) and +(0,-1) .. (1,3);
\shadedraw [black,ball color=gray] (3.2,3) .. controls +(0,2) and +(0,2) .. (4.8,3) .. controls +(0,-2) and +(0,-2) .. (3.2,3);
\shadedraw [black,ball color=gray] (2.5,6) .. controls +(0,1.5) and +(0,1.5) .. (5.5,6) .. controls +(0,-1.5) and +(0,-1.5) .. (2.5,6);
\shadedraw [black,ball color=gray] (6.5,6) .. controls +(0,2.5) and +(0,2.5) .. (9.5,6) .. controls +(0,-2.5) and +(0,-2.5) .. (6.5,6);

\draw [line width=1mm] (0,0) rectangle (2,3);
\draw [line width=1mm] (2,0) rectangle (4,3);
\draw [line width=1mm] (0,3) rectangle (4,6);
\draw [line width=1mm] (4,0) rectangle (8,6);
\draw [line width=1mm] (4,0) rectangle (8,6);
\draw [line width=1mm] (0,6) rectangle (8,10);
\draw [line width=1mm] (8,0) rectangle (16,10);

\draw (0.7,1.4) node {$D_0$};
\draw (1.65,1.4) node {$D_1$};
\draw (1.8,3.4) node {$D_2$};
\draw (3.55,3.4) node {$D_3$};
\draw (3.6,6.4) node {$D_4$};
\draw (7.5,6.4) node {$D_5$};

\draw (7.2,9) node {$A^{(4)}$};
\draw (8.8,9) node {$B^{(4)}$};
\draw (6,5.5) node {$A^{(3)}$};
\draw (6,6.5) node {$B^{(3)}$};
\draw (3.3,4.6) node {$A^{(2)}$};
\draw (5,4.6) node {$B^{(2)}$};
\draw (0.6,2.5) node {$A^{(1)}$};
\draw (0.6,3.5) node {$B^{(1)}$};
\draw (3,0.4) node {$B^{(0)}$};
\draw (1.3,0.4) node {$A^{(0)}$};

\end{tikzpicture}
\caption{A bisection sequence $(A^{(j)}, B^{(j)})$, with the distributions $D_j$ returned by the mechanism at level $j$. The density of $D_j$ increases in a certain technical sense exponentially in $j$, although much slower than $2^j$.}
\label{fig:bisect-seq}
\end{figure}

At level $j$, $1 \leq j \leq \ell$, we have a (random) set $A^{(j)}$ of $m / 2^{\ell-j}$ items, which is partitioned randomly into two sets $A^{(j-1)} \cup B^{(j-1)}$ of equal size; these are level-$(j-1)$ sets.
The valuation function at level $j$ will be as in Section~\ref{sec:exp-gap} but restricted to the set $A^{(j)} = A^{(j-1)} \cup B^{(j-1)}$ (the two parts play the role of $M_1,M_2$ from Section~\ref{sec:exp-gap}). The mechanism can detect the set $A^{(j)}$; however, the partition of $A^{(j)}$ into $A^{(j-1)} \cup B^{(j-1)}$ remains hidden. By the symmetry gap argument, the mechanism cannot learn what the partition is, and hence any distribution $D_j$ returned by the algorithm will be with high probability balanced with respect to $(A^{(j-1)}, B^{(j-1)})$. The MIDR property implies that this distribution must be ``dense" enough in order to beat the distribution $D_{j-1}$ guaranteed by the previous level, which is sensitive to the partition $(A^{(j-1)}, B^{(j-1)})$. (By density, we mean a certain notion of average size for sets sampled from $D_j$.) Since distributions concentrated inside $A^{(j-1)}$ or $B^{(j-1)}$ are more profitable than distributions balanced between $(A^{(j-1)}, B^{(j-1)})$, we will ideally obtain a constant-factor boost in density at each level. As $\ell$ grows, this will eventually contradict the fact that the mechanism cannot choose more than $k$ items. 

Finding the right definition of density that yields a constant-factor boost at each level is the main technical difficulty. The most natural definition of density seems to be the expected size of the set returned by the mechanism. However, this notion does not yield the desired boost. (This is related to the fact that we cannot get any contradiction for coverage functions.) The notion of density that turns out to be useful is more complicated; it is derived from functions that exhibit non-concave behavior of the extension $F^{exp}$. 
This strategy will be made more explicit in the following.


\paragraph{The symmetry gap.}
At level $j+1$, we consider valuation functions of the form
$$ f_{A^{(j)},B^{(j)}}(S) = 1 - \left( 1 - \phi\left(\frac{|S \cap A^{(j)}|}{|A^{(j)}|}\right) \right)
 \left( 1 - \phi\left(\frac{|S \cap B^{(j)}|}{|B^{(j)}|}\right) \right) $$
where $\phi:[0,1] \rightarrow [0,1]$ is a suitable non-decreasing concave function. Note that under this valuation function, the value of a (random) set $R$ depends only on how many elements it takes from $A^{(j)}$ and $B^{(j)}$.
In particular, if we denote $X_j = \frac{|R \cap A^{(j)}|}{|A^{(j)}|}$,
$Y_j = \frac{|R \cap B^{(j)}|}{|B^{(j)}|}$, then we have
$$ \E[f_{A^{(j)},B^{(j)}}(R)] = \E[1 - (1-\phi(X_j))(1-\phi(Y_j))].$$
Since the expected value depends only on $X_j,Y_j$, we say that the random variables $X_j, Y_j$ represent the distribution of $R$.

By the symmetry gap argument (if the valuation function is suitably perturbed and the partition $(A^{(j)},B^{(j)})$ is random), then the mechanism with high probability returns a solution $R^{(j+1)}$ independent of the partition and hence symmetric with respect to it. Denoting
$X_{j+1} = \frac{|R^{(j+1)} \cap A^{(j+1)}|}{|A^{(j+1)}|}$, we obtain that the mechanism returns expected value
$$ \E[f_{A^{(j)},B^{(j)}}(R^{(j+1)})] = \E[1 - (1-\phi(X_{j+1}))^2] $$
which is typically less than $\E[1 - (1-\phi(X_j))(1-\phi(Y_j))]$ if $X_{j+1} = \frac12 (X_j + Y_j)$
and $X_j \neq Y_j$.
(There are certain error terms arising from the symmetry gap argument but let us ignore them for now.)

The main point here is that if the mechanism is MIDR, then the expected value of the returned random set $\E[f_{A^{(j)},B^{(j)}}(R^{(j+1)})]$ must be at least that of any other random set whose distribution is in the range - in particular, the random set $R^{(j)}$ whose presence in the range we prove at the previous level. If this random set $R^{(j)}$ is represented by the random variables $X_j, Y_j$, then the mechanism must return a distribution represented by $X_{j+1}$ such that
$$ \E[1 - (1-\phi(X_{j+1}))^2] \geq \E[1 - (1-\phi(X_j))(1-\phi(Y_j))].$$
We in fact ignore the contribution of $Y_j$ and use the weaker inequality
\begin{equation}
\label{eq:sym-simple}
\E[1 - (1-\phi(X_{j+1}))^2] \geq \E[\phi(X_j)].
\end{equation}
Hence, the existence of certain distributions in the range forces the existence of other distributions, satisfying the bound (\ref{eq:sym-simple}). 

\paragraph{Gap amplification.}
Now we would like to say that if two distributions represented by $X_j$ and $X_{j+1}$ satisfy (\ref{eq:sym-simple}),
then the distribution at level $j+1$ is ``more dense" than the one at level $j$.
Considering the scaling at different levels, we want to prove that $X_{j+1}$ is ``significantly larger" than $\frac12 X_j$.
This is intuitive, since $X_{j+1} = \frac12 X_j$ is not enough to satisfy (\ref{eq:sym-simple}), for example when $\phi$ is linear.
Unfortunately,  (\ref{eq:sym-simple}) does not imply any useful relationship between the expectations $\E[X_j]$, $\E[X_{j+1}]$, beyond $\E[X_{j+1}] \geq \frac12 \E[X_j]$.
For example, we could have $X_j = 1$ with probability $\xi - \xi^2$ and $0$ otherwise. Then $X_{j+1} = \frac12 \xi$ satisfies (\ref{eq:sym-simple}) for any concave function $\phi:[0,1] \rightarrow [0,1]$. This does not provide a constant-factor improvement over $\frac12 \E[X_j]$.

We still want to prove that $X_{j+1}$ is in some sense ``significantly larger" than $\frac12 X_j$.
Our main technical inequality formalizing this intuition is the following: 
Define $\phi_\alpha(t) = \min \left\{ \frac{t}{\alpha}, 1 \right\}$.
Then for any distribution in the range represented by $X_j$ at level $j$
and any $\alpha_j \in [0,1]$, there is a distribution in the range represented by $X_{j+1}$ at level $j+1$, and $\alpha_{j+1} \in [0,1]$
such that
\begin{equation}
\label{eq:ampl}
\alpha_{j+1} (\E[\phi_{\alpha_{j+1}}(X_{j+1})])^{1+\delta} \geq \left( \frac{1+\delta^2}{2} \right) \alpha_{j} (\E[\phi_{\alpha_{j}}(X_{j})])^{1+\delta} .
\end{equation}
where $\delta>0$ is some (small) absolute constant. The use of $1+\delta$ in the exponent is crucial here.
We remark that formulating and proving this inequality was the most challenging part of the proof.
(The precise statement with a proof appears as Lemma~\ref{lem:level-bound} in the appendix.)

\paragraph{The contradiction.}
Using this bound, we arrive at a contradiction as follows. As we already mentioned, assuming that an MIDR mechanism provides a $c$-approximation for the CPP problem, then for any feasible set $A^{(0)}$ there must be a distribution $D_0$ in its range such that
$$ \E[X_0] = \E_{R \sim D_0}\left[\frac{|R \cap A^{(0)}|}{|A^{(0)}|}\right] \geq c.$$
Now we apply the symmetry gap argument and the gap amplification technique to random pairs of sets $(A^{(j)}, B^{(j)})$
at each level $j$. Starting from $\E[X_0] \geq c$ and $\alpha_0=1$, by repeated use of (\ref{eq:ampl})
we obtain that there is $\alpha_\ell \in [0,1]$ and a distribution at level $\ell$ represented by $X_\ell$ such that
$$ \alpha_\ell (\E[\phi_{\alpha_\ell}(X_\ell)])^{1+\delta} \geq \left( \frac{1+\delta^2}{2} \right)^\ell c^{1+\delta}. $$
Note that $ \alpha_\ell (\E[\phi_{\alpha_\ell}(X_\ell)])^{1+\delta} \leq \alpha_\ell \E[\phi_{\alpha_\ell}(X_\ell)] 
 = \E[\min \{X_\ell, \alpha_\ell\}] \leq \E[X_\ell]$. So in fact
$$ \E[X_\ell] \geq \left( \frac{1+\delta^2}{2} \right)^\ell c^{1+\delta} > \frac{2^{\delta^2 \ell}}{2^\ell} c^{1+\delta}. $$
The meaning of $X_\ell$ 
is simply the fraction of the ground set that the mechanism returns at level $\ell$. Since $m = 2^{O(\ell)}$, we have $2^{\delta^2 \ell} \geq m^{(1+\delta)\gamma}$ for some constant $\gamma>0$.
If the approximation factor is $c \geq m^{-\gamma}$, then we get $\E[X_\ell] > 2^{-\ell}$,
which would violate the cardinality constraint of the CPP problem.

As we mentioned, the full proof appears in Appendix~\ref{sec:CPP-proof}.


\section{Hardness for combinatorial auctions}
\label{sec:auctions-hardness}

The following is our main result for combinatorial auctions.

\begin{theorem}
\label{thm:CA-hardness}
There are absolute constants $\epsilon, \gamma > 0$ such that there is no $(1-\epsilon)$-approximately truthful-in-expectation mechanism for combinatorial auctions with monotone submodular valuation functions in the value oracle model, achieving a better than $1/n^\gamma$-approximation in expectation in terms of social welfare, where the number of players is $n$ and the number of items is $m=\poly(n)$. 
\end{theorem}

\paragraph{Discussion.}
This theorem extends previous negative results for combinatorial auctions with submodular valuation functions, which were known in the cases of deterministic truthful and randomized universally truthful mechanisms \cite{Dobzin11}. Also, it appears that as stated these results do not rule out approximately truthful mechanisms.

We remark that there is still the possibility of a truthful-in-expectation mechanism in the ``lottery-value" oracle model which was introduced in \cite{DRY11}. Here, a player is able to provide the {\em exact} expectation $\E[v_i(\hat{\bx})]$ for a product distribution given by $\bx$. Since the exact expectations $\E[v_i(\hat{\bx})]$ are hard to compute even in very special cases like the budget-additive case, this is a severe limitation. Our hardness result does not apply directly to this stronger oracle model. However, what our result implies is that if a truthful-in-expectation mechanism exists in the lottery-value model, then it must be very sensitive to the accuracy of the oracle's answers, and does not remain even approximately truthful-in-expectation if the oracle's answers involve some small noise. This is because if we had a mechanism in the lottery-value oracle model, which remains approximately t.i.e. under small noise in the oracle and provides a good approximation, then we could simulate this mechanism in the value oracle model (by sample-average approximation). Thus we would obtain an approximately t.i.e. mechanism contradicting Theorem~\ref{thm:CA-hardness}. 

\paragraph{Proof strategy.} Our hardness result for combinatorial public projects (Section~\ref{sec:CPP-hardness}) can be adapted to show that there is no (approximately) MIDR mechanism for submodular combinatorial auctions that guarantees a good approximation ratio. However, unlike in CPP, we are unable to prove that truthful-in-expectation mechanisms and MIDR algorithms are equivalent in power (even in the approximate sense). This is not surprising, since randomized truthful mechanisms that are not maximal-in-distributional-range have been designed for combinatorial auctions (see for example \cite{D07}). Therefore, additional ideas are needed to rule out all truthful-in-expectation mechanisms.
Such ideas have been recently put forth in a paper by Dobzinski \cite{Dobzin11}. The {\em direct hardness} approach of \cite{Dobzin11} provides a way to avoid the characterization step and instead attack the truthful mechanism directly. This idea applies to truthful-in-expectation mechanisms as well. 

The main idea of the direct hardness approach can be stated as follows. If we identify a special player whose range of possible allocations is sufficiently ``rich" when the valuations of other players are fixed to particular functions, then we can work with the special player directly using the \emph{taxation principle}: There is a fixed price for each distribution over allocations in the ``range" of the mechanism as the special player varies his valuation, and  the mechanism outputs the distribution in this range that maximizes the player's utility (his expected value for the distribution on allocations less the price of that distribution). Thus, our symmetry gap techniques from Section~\ref{sec:CPP-hardness} apply here quite naturally, though the presence of payments poses an additional technical challenge that was not present for CPP.
Next, we present a sketch of our proof. The complete proof is presented in Appendix~\ref{sec:CA-proof}.

\paragraph{The basic instance.}
We start from the following ``basic instance".
For an integer $\ell$, we construct instances with $|N| = n = 2^\ell$ players and $|M| = m = \poly(n)$ items.
Each player has a ``polar valuation" $v^*_i$ (as in \cite{Dobzin11}), where items in a certain set $A^{(0)}_i$ have value $1$ for player $i$ and other items have (small) value $\omega>0$. The sets $A^{(0)}_i$ are chosen independently at random, under the constraint that $|A^{(0)}_i| = m/n$.

A counting argument shows that if a mechanism provides a $c$-approximation in social welfare,
then there must be a player whose allocated set $R^{(0)}_i$ overlaps significantly with his desired set $A^{(0)}_i$:
\begin{equation*}
 \E[|R^{(0)}_i \cap A^{(0)}_i|] \geq (c/4-\omega) \E[|R^{(0)}_i \cup A^{(0)}_i|].
\end{equation*}
(See Lemma~\ref{lem:basic-inst}.)
By an averaging argument, this is also true for a certain fixed choice of the other players' valuations. In the following, we fix that choice and consider varying valuations for player $i$ only, who we refer to as the ``special player". We also drop the index $i$, since we do not consider the other players anymore.

In the following, we set $\omega = c/8$, so that $\E[|R^{(0)} \cap A^{(0)}|] \geq \omega \E[|R^{(0)} \cup A^{(0)}|]$. Hence we can estimate the expected value received by the special player as follows: 
$$\E[v^*(R^{(0)})] = \E[|R^{(0)} \cap A^{(0)}|] + \omega \E[|R^{(0)} \setminus A^{(0)}|] \leq 2 \E[|R^{(0)} \cap A^{(0)}|].$$
Denoting $X_0 = \frac{|R^{(0)} \cap A^{(0)}|}{|A^{(0)}|}$, we have  $\E[v^*(R^{(0)})] \leq \frac{2m}{n} \E[X_0]$.
Also, $\E[v^*(R^{(0)})] \geq \E[|R^{(0)} \cap A^{(0)}|] = \frac{m}{n} \E[X_0]$. 
So the special player's utility in the basic instance (in expectation over the random instances) is $\frac{m}{n} \E[X_0]$, up to a factor of $2$.

\paragraph{Symmetry gap again.}
We consider valuations for the special player at $\ell$ levels, in the same form that we considered in the case of combinatorial public projects. The difference now is that the mechanism is not necessarily maximal-in-distributional-range. Instead, we use the definition of truthfulness in expectation directly. The same symmetry gap argument as in Section~\ref{sec:CPP-hardness} gives the following:
If there is a random set $R^{(j)}$ possibly allocated at level $j$ at a price $P_j$, and $X_j = \frac{|R^{(j)} \cap A^{(j)}|}{|A^{(j)}|}$, then there is a random set $R^{(j+1)}$ possibly allocated at level $j+1$ at a price $P_{j+1}$, and
$X_{j+1} = \frac{|R^{(j+1)} \cap A^{(j+1)}|}{|A^{(j+1)}|}$, so that
$$  \E[1 - (1-\phi(X_{j+1}))^2] - \E[P_{j+1}] \geq \E[\phi(X_j)] - \E[P_j].$$
Again, we are ignoring certain error terms and we are also ignoring the issue of approximate truthfulness.
Using the fact that the valuation functions can be scaled arbitrarily and the mechanism must still be truthful in expectation,
we obtain that for any $\lambda',\lambda'' \geq 0$, there is distribution possibly allocated at level $j+1$ such that
\begin{equation}
\label{eq:sym-lam}
 \lambda' \E[1 - (1-\phi(X_{j+1}))^2] - \lambda'' \E[P_{j+1}] \geq \lambda' \E[\phi(X_j)] - \lambda'' \E[P_j].
\end{equation}

\paragraph{Convex hulls and the separation argument.}
Our goal is to eliminate the prices from the picture, so that we can use arguments similar to Section~\ref{sec:CPP-hardness}. For that purpose, it is convenient to pass to convex hulls as follows. We define the {\em distribution menu} $\cM_j$ at level $j$ to consist of all distributions of pairs of random variables $(X_j,P_j)$, such that $X_j = \frac{|R^{(j)} \cap A^{(j)}|}{|A^{(j)}|}$ for some random set $R^{(j)}$ allocated for a level-$j$ valuation at a price $P_j$. Then we define the {\em closure} of a distribution menu, $\overline{\cM}_j$, to be the topological closure of the convex hull of $\cM_j$ (in the sense of taking convex combinations of distributions). By convexity, (\ref{eq:sym-lam}) still holds in the sense that for any $(X_j,P_j)$ with a distribution in $\overline{\cM}_j$ and any $\lambda',\lambda'' \geq 0$, there is $(X_{j+1},P_{j+1})$ with a distribution in $\overline{\cM_{j+1}}$ such that (\ref{eq:sym-lam}) holds.

A convex separation argument, essentially Farkas' lemma in 2 dimensions, actually implies the following.
For any $(X_j,P_j)$ with a distribution in $\overline{\cM}_j$, there is $(X_{j+1},P_{j+1})$ with a distribution in $\overline{\cM_{j+1}}$ such that $\E[P_{j+1}] \leq \E[P_j]$ and
\begin{equation*}
 \E[1 - (1-\phi(X_{j+1}))^2] \geq \E[\phi(X_j)].
\end{equation*}
In other words, there is a distribution in the closure of the menu at level $j+1$ at a price no higher than the price we had at level $j$, and the respective random variables $X_j, X_{j+1}$ satisfy the same relationship (\ref{eq:sym-simple}) that we had in Section~\ref{sec:CPP-hardness}.
The rest of the proof goes exactly as in Section~\ref{sec:CPP-hardness}, using (\ref{eq:ampl}) and eventually producing a distribution represented by $(X_\ell,P_\ell)$ in $\overline{\cM_\ell}$ such that $\E[P_\ell] \leq \E[P_0]$ and
$$ \E[X_\ell] \geq \left( \frac{1+\delta^2}{2} \right)^\ell (\E[X_0])^{1+\delta}. $$
Here, $(X_0,P_0)$ represents the distribution and price allocated in the basic instance.
Now we consider the utility that the distribution represented by $(X_\ell,P_\ell)$ would provide in the basic instance: since every element has value at least $\omega$ there, the utility would be
\begin{equation}
\label{eq:ell-value}
 \E[v^*(R^{(\ell)}) - P_\ell] \geq \E[\omega m X_\ell - P_\ell] \geq \omega m \left( \frac{1+\delta^2}{2} \right)^\ell (\E[X_0])^{1+\delta} - \E[P_0].
\end{equation}
Recall that the distribution of $(X_\ell,P_\ell)$ is not on the menu $\cM_\ell$ but rather in its convex hull. However, by using the properties of the convex hull, there must be a distribution on the actual menu $\cM_\ell$ that satisfies the same linear inequality. So we can assume without loss of generality that the distribution of $(X_\ell, P_\ell)$ is on the actual menu at level $\ell$, and $R^{(\ell)}$ is the respective random set that would be allocated to the special player if he declared a level-$\ell$ valuation.

Recall that in the basic instance, the value received by the special player is at most $\frac{2m}{n} \E[X_0]$, and the respective utility is at most $\frac{2m}{n} \E[X_0] - \E[P_0]$. We also have $n = 2^\ell$ and $\E[X_0] \geq c/4 - \omega = c/8$.
If $c = 8\omega \geq n^{-\gamma}$ for a suitable constant $\gamma>0$, we would obtain from (\ref{eq:ell-value}) that the special player could substantially improve his utility in the basic instance by declaring a level-$\ell$ valuation instead. We conclude that this would contradict the property of truthfulness in expectation.

The complete proof appears in Appendix~\ref{sec:CA-proof}.


\appendix

\section{Transforming an approximately truthful-in-expectation mechanism into approximately maximal-in-distribution range}
\label{sec:TIE-MIDR}
In this section, we prove Theorem \ref{thm:TIE=MIDR}.

%
%
Let $\Omega= \set{ S \sse [m] : |S| \leq k}$ be the set of outcomes of CPP. Let $\Delta(\Omega)$ denote the simplex in $\RR^\Omega$, representing the set of distributions over $\Omega$. Let $\V$ denote the set of submodular valuations on $[m]$. We think of $\V$ as a subset of $\RRp^\Omega$ --- specifically, each $v \in \V$ is a vector in $\RRp^\Omega$, where $v_S$ is the value of outcome $S$ for a player with valuation $v$. We note that for each $v \in \V$,  the infinity norm $||v||_\infty$  is equal to the value of the optimum solution.

Fix $\epsilon$, $c$, $\cM$, and $\delta$ as in the statement of the theorem. Let $\A: \V \to \Delta(\Omega)$ be the allocation rule of $\cM$ when there is a single player. By assumption, $\A$ is a $c$-approximation for $c>0$ -- specifically, $\frac{v^T \A(v)}{||v||_\infty} \geq c$.  The following is an approximate variant of \emph{weak monotonicity} \cite{LMN03}, and follows from the fact that $\A$ is the allocation rule of a  $(1-\epsilon)$-approximately truthful mechanism.

\begin{fact}[Similar to \cite{LMN03}]\label{fact:wmon}
  For any $u,v \in \V$,
\[ v^T \A(v) - (1-\epsilon) u^T \A(v) \geq (1-\epsilon) v^T \A(u) - u^T \A(u).  \]
\end{fact}

We prove Theorem \ref{thm:TIE=MIDR} by showing that there is a black-box reduction  that converts $\A$ to a new allocation rule $\B$ that is $(1- 3\epsilon - \delta)$-MIDR. The reduction will be non-uniform -- specifically, $\B$ will utilize an advice string that depends on $m$, but is independent of the input valuation $v \in \V$. The length of the advice string will not be bounded, polynomially or otherwise --- this is OK, since we are only interested in preserving value oracle lower-bounds. $\B$ preserves  the approximation ratio of $\A$, and moreover makes only $m$ more value queries than does $\A$. 



The proof consists of two main steps. First, we show that $\A$ tends to a $(1-\epsilon)$-approximately maximal-in-distributional-range allocation rule ``in the limit'' as we scale up the valuations.  Then, we use this fact to construct, via a non-uniform black box reduction,  an allocation rule $\B$ that approximates the limit behavior of $\A$,  in the sense that it $(1-3\epsilon -\delta)$-approximately maximizes over the range of $\A$. 

\begin{remark}
We note that the proofs of this section apply more generally than CPP with submodular valuations. In particular, the only properties of this problem that are used in the proofs are: (1) The multiple-player allocation problem is algorithmically equivalent to the single player allocation problem (2) The set $\Omega$ of outcomes is finite, (3) The set of valuations $\V \sse \RRp^\Omega$ is closed under scaling by a non-negative constant, and (4) There is a deterministic algorithm $s: \V \to \RRp$ that runs in finite time, makes a polynomial number of value queries, and returns a ``weak approximation'' to the optimal value -- i.e. we only require that $s(v) > 0$ when $||v||_\infty >0$. When a welfare-maximization mechanism design problem satisfies these four conditions, as do all variants of CPP and other ``public-project''-type problems in the literature, then the analogue of Theorem \ref{thm:TIE=MIDR} holds for that problem.
\end{remark}

\subsection{Limit behavior of truthful in expectation mechanisms}

We will show that $\A$ is $(1-\epsilon)$-approximately MIDR in the limit as we scale up the valuations. Recall that a mechanism $\B: \V \to \Delta (\Omega)$ is $(1-\epsilon)$-approximately MIDR if $v^T\B(v) \geq (1-\epsilon) \sup_{w \in \V} v^T \B(w)$ for all $v \in \V$. The following statement is analogous.

\begin{proposition}\label{prop:lim_midr}
$\liminf_{\alpha \to \infty} v^T\A(\alpha v) \geq (1-\epsilon) \sup_{w \in \V} v^T\A(w)$ for all $v \in \V$.
\end{proposition}
\begin{proof}
Let $\alpha,\beta\in \RRp$, and let $w,v \in \V$. By Fact \ref{fact:wmon}:
\begin{align*}
 \alpha v^T \A(\alpha v) -   (1-\epsilon)w^T \A(\alpha v) &\geq (1-\epsilon) \alpha v^T \A( w) -   w^T \A( w).
\end{align*}
Dividing the expression by $\alpha$, we get:
\begin{align*}
 v^T \A(\alpha v) -  \frac{ (1-\epsilon) w^T \A(\alpha v)}{\alpha} &\geq  (1-\epsilon) v^T \A( w) -  \frac{ w^T \A( w)}{\alpha}
\end{align*}
Taking the limit infimum as $\alpha$ goes to infinity,
\begin{align*}
\liminf_{\alpha \to \infty} v^T \A(\alpha v)  &\geq  (1-\epsilon) v^T \A( w) .
\end{align*}
Now, taking the supremum over $w$
\begin{align*}
\liminf_{\alpha \to \infty} v^T \A(\alpha v)  &\geq (1-\epsilon) \sup_{w \in \V}  v^T \A(w) 
\end{align*}
 This  completes the proof.
\end{proof}


 

\subsection{Approximating the limit behavior of a mechanism}
Ideally, we would transform $\A$ to an allocation rule that behaves as $\A$ does in the limit -- by the results of the previous sub-section, such a ``limit allocation rule'' of $\A$ would be $(1-\epsilon)$-MIDR. However, since our reduction must take finite time, we must settle for approximating the limit behavior of $\A$. Unfortunately, even that is non-trivial: given $v$, the ratio $\alpha$ by which we would need to scale $v$ before coming close to the ``limit'' of $\A(\alpha v)$ is a complete mystery, and may be arbitrarily large. Therefore, we need to utilize some non-uniform advice to deduce that order of magnitude of the necessary scaling factor. An additional difficulty is that this advice must be independent of $v$ -- specifically,  the advice may depend only on the number of items $m$.

For each $\delta' >0$ and $v \in \V$, we define a threshold  $t(\delta',v)$. Roughly speaking, $t(\delta',v)$ is the ``scale''  at which  $\A$ is guaranteed to be within $(1-\delta')$ of its limit behavior when given input in the direction of $v$.  Proposition \ref{prop:lim_midr} guarantees that threshold $t(\delta',v)$ exists for each $v\in \V$ and $\delta' >0$.
\begin{equation}
  \label{eq:2}
  t(\delta',v) = \sup \left\{t :  v^T \A\left(t \frac{v}{||v||_\infty}\right) \leq (1- \epsilon-\delta') \sup_{w \in \V} v^T \A(w)\right\} +1
\end{equation}
 We note that the motivation for adding $1$ (any arbitrary positive number would do) to the expression is to guarantee that $v^T \A\left(t(\delta',v) \frac{v}{||v||_\infty}\right) \geq (1-\epsilon - \delta') \sup_{w \in \V} v^T \A(w)$, which may not be guaranteed by the supremum.

Assume that, for some $\delta'>0$, we have some upper-bound $\tau$ on $\set{t(\delta',v): v\in \V}$, and moreover we have a procedure $s(v)$ to estimate a non-zero lower bound on $||v||_\infty$ for each $v \in \V$. Then, the following procedure is evidently a $c$-approximate and $(1-\epsilon - \delta')$-MIDR allocation rule: Given input $v \in \V$, output  $\A(\frac{\tau}{s(v)} \cdot v)$. The procedure $s(v)$ is easy to implement using only $m$ value queries --- indeed, we can take $s(v) = \max_{j \in [m]} v(\set{j})$. It is not clear, however, that the  upper-bound $\tau$ can be computed effectively. Even worse, it is not clear that such an upper-bound even exists: $\V$ is infinite, and $t(\delta',v)$ is not necessarily a continuous function of $v$!

We remedy this as follows. We will show that there exists such an upper-bound when $\delta'$ is sufficiently large relative to $\epsilon$. Specifically, we show that there exists an  upper-bound $\tau$ on  $\set{t(\delta + 2\epsilon,v) : v \in \V}$. However, computing such an upper-bound in finite time may be impossible in general, since the scale of valuations at which $\A$ approaches its  limit behavior may be arbitrary. Instead, we take $\tau$ as advice to our non-uniform reduction. By the discussion in the previous paragraph, showing that the upper-bound $\tau$ exists  yields a non-uniform allocation rule that is $(1-3 \epsilon -\delta)$-MIDR, and makes at most $m$ more value queries than $\A$, completing the proof of Theorem \ref{thm:TIE=MIDR}. 

As a tool  for proving that the upper-bound $\tau$ exists, we define a finite net of $\V$. Since $\V$ is a  cone  in finite-dimensional euclidean space, its intersection with the infinity-norm unit ball admits a \emph{$\sigma$-net} in the infinity-norm for any $\sigma>0$  --- specifically, a finite set  $\U \sse \V$ such that
\begin{enumerate}
\item $||u||_\infty =1$ for all $u \in \U$
\item  $\forall v \in \V \  \exists u \in \U \ \left|\left|\frac{v}{||v||_\infty} - u\right|\right|_\infty \leq \sigma $
\end{enumerate}

Let $\sigma = c \delta / 4$ and let $\U$ be a $\sigma$-net of $\V$. Now let $\beta= \max_{u \in \U} t(\frac{\delta}{4},u)$, and let  $\tau= 4 \beta / \delta$. It suffices to show that $v^T \A ( \tau v) \geq (1-3\epsilon -\delta) \sup_{w \in \V} v^T \A (w)$ for each  $v \in \V$ with $||v||_\infty = 1$. Let $v \in \V$ be such that $||v||_\infty = 1$, and let $u$ be a point in the $\sigma$-net $\U$ such that $||v- u||_\infty \leq \sigma$.  By Fact \ref{fact:wmon}, we have

\begin{align}\label{eq:wmon_tau}
 \tau v^T \A(\tau v ) - (1-\epsilon) \beta u^T \A(\tau v) &\geq (1-\epsilon)\tau v^T \A (\beta u)  - \beta u^T \A( \beta u)    
\end{align}
We now use inequality \eqref{eq:wmon_tau} to lower-bound $v^T \A(\tau v)$:
\begin{align*}
  v^T \A(\tau v)  &\geq  (1-\epsilon) v^T \A (\beta u )  - \frac{ \beta}{\tau} u^T \A( \beta u)    &&\text{Dividing \eqref{eq:wmon_tau} by $\tau$ and loosening the inequality}\\
   &=  (1-\epsilon) v^T \A (\beta u)  - \frac{\delta}{4} u^T \A( \beta u)    &&\text{By definition of $\tau$}\\
   &=  \left(1  - \epsilon - \frac{\delta}{4}\right) u^T \A( \beta u) - (1-\epsilon) (u-v)^T \A(\beta u)    &&\\
   &\geq  \left(1- \epsilon - \frac{\delta}{4}\right) u^T \A( \beta u) - ||u-v||_\infty     &&\text{Since $||\A(\beta u)||_1 = 1$} \\
   &\geq  \left(1  - \epsilon-\frac{\delta}{4}\right) u^T \A( \beta u) - \sigma     &&\text{By proximity of $u$ and $v$} \\
   &\geq  \left(1 - \epsilon -\frac{\delta}{2}\right) u^T \A( \beta u)     &&\text{By  $c \leq u^T \A(\beta u)$ and definition of $\sigma$} \\
  &\geq  \left(1-2\epsilon -\frac{3}{4} \delta\right) \sup_{w \in \V} u^T \A (w)      &&\text{By definition of $\beta$}\\
  &\geq  \left(1-2 \epsilon -\frac{3}{4} \delta\right) \liminf_{\alpha \to \infty} u^T \A (\alpha v)      &&\\
  &=  \left(1-2 \epsilon -\frac{3}{4} \delta\right) \liminf_{\alpha\to \infty} (v^T \A (\alpha v) - (v -u)^T \A(\alpha v) )      &&\\
  &\geq  \left(1- 2 \epsilon -\frac{3}{4} \delta\right) \liminf_{\alpha \to \infty} (v^T \A (\alpha v) - ||v -u||_\infty  )      &&\text{Since $||\A(\alpha v)||_1 = 1$}\\
  &\geq  \left(1-2 \epsilon - \frac{3}{4} \delta\right) \liminf_{\alpha \to \infty} (v^T \A (\alpha v) - \sigma  )      &&\text{By proximity of $u$ and $v$}\\
  &\geq  \left(1- 2 \epsilon - \frac{3}{4}\delta\right) \liminf_{\alpha \to \infty} (v^T \A (\alpha v) - \frac{\delta}{4} v^T \A(\alpha v)  )      &&\text{By  $c \leq v^T \A(\alpha v)$ and definition of $\sigma$}\\
  &\geq  \left(1- 2 \epsilon - \delta\right) \liminf_{\alpha \to \infty} v^T \A (\alpha v)  &&\\
  &\geq  \left(1- 3 \epsilon - \delta\right) \sup_{w \in \V} v^T \A (w)  &&\text{By Proposition \ref{prop:lim_midr}}
\end{align*}
By the previous discussion, this completes the proof of Theorem \ref{thm:TIE=MIDR}.






\section{Proof of hardness for combinatorial public projects}
\label{sec:CPP-proof}

In this section, we present the proof of Theorem~\ref{thm:CPP-hardness}.

\subsection{The basic setup}
\label{sec:basics}

We consider a ground set of $|M|=m=400^\ell$ items for some $\ell \geq 1$. We set the cardinality bound to be $k = 200^\ell = m/n$ where $n = 2^\ell$. (We note that $n$ is just a parameter unrelated to the number of players, which is 1 in this case. This parameter will however denote the number of players in Section~\ref{sec:auctions-hardness}.)
An important object in the following will be a {\em random bisection sequence}.

\begin{definition}
\label{def:rnd-bisect}
A random bisection sequence is a random sequence of pairs of sets $(A^{(0)},B^{(0)})$, $(A^{(1)},B^{(1)})$,
$\ldots$, $(A^{(\ell)}, B^{(\ell)})$ generated as follows. We define $A^{(\ell)} = B^{(\ell)} = M$. Given $A^{(j)}$ for $0 < j \leq \ell$, we pick $(A^{(j-1)}, B^{(j-1)})$ uniformly among all partitions of $A^{(j)}$ into two parts of size $\frac12 |A^{(j)}|$.
\end{definition}

I.e., $|A^{(j)}| = |B^{(j)}| = 2^{j-\ell} m$.
We refer to $A^{(j)} = A^{(j-1)} \cup B^{(j-1)}$ as the $j$-th level of the bisection sequence. Observe that the distribution of $(A^{(j-1)},B^{(j-1)})$ is uniform among all pairs of disjoint sets of size $2^{j-1-\ell} m$.
We will use valuation functions associated with each level of a bisection sequence. We denote these valuation functions at level $j$ by $f_{A^{(j-1)},B^{(j-1)}}$. In particular, this valuation function depends only on the elements of $A^{(j)} = A^{(j-1)} \cup B^{(j-1)}$.

The bisection sequence is generated at random and unknown to the mechanism. For each particular choice of a valuation function at a certain level, the mechanism needs to produce a probability distribution over feasible sets, which is purportedly the (approximately) optimal one over a certain fixed range of distributions $\cR$. The distribution will depend on the choice of a valuation function, in particular on the relevant set of items $A^{(j)}$. A function assigning a distribution over sets to every set $A^{(j)}$ is a complicated object; in order to be able to argue about all possible such functions, we distill the important information into a single random variable for each level $j$.

\begin{definition}
\label{def:density}
We say that a random variable $X_j$ is constructible by a range $\cR$ at level $j$, if there is a distribution $D(A^{(j)}) \in \cR$ for each set $A^{(j)}$ of size $2^{j-\ell} m$ such that if a random set $R$ is generated by first choosing $A^{(j)}$ uniformly among all sets of size $2^{j-\ell} m$ and then sampling $R$ from the distribution $D(A^{(j)})$, then
$$ X_j = \frac{|R \cap A^{(j)}|}{|A^{(j)}|}. $$
\end{definition}

Note that the normalization is chosen so that we have $X_j \in [0,1]$. There are two sources of randomness in defining $X_j$: one is the randomness in $A^{(j)}$, and one arises from the probability distribution $D(A^{(j)})$.

The first useful fact is the following (easy) lemma.

\begin{lemma}
\label{lem:base-case}
Consider a mechanism returning distributions from a range $\cR$ that achieves a $c$-approximation for the problem $\max \{f(S): |S| \leq k \}$ for $f$ monotone submodular. Then there is a random variable $X_0$ constructible by $\cR$ at level $0$ such that
$$ \E[X_0] \geq c.$$
\end{lemma}

\begin{proof}
Consider the valuation function $f_{A^{(0)}}(S) = \frac{|S \cap A^{(0)}|}{|A^{(0)}|}$. Let $X_0 = \frac{|R^{(0)} \cap A^{(0)}|}{|A^{(0)}|}$ where $R^{(0)}$ is the random set returned by the mechanism, given valuation $f_{A^{(0)}}$ for $A^{(0)}$ chosen randomly among all sets of size $2^{-\ell} m$. By Definition~\ref{def:density}, $X_0$ is a random variable constructible by the range $\cR$ at level $0$. 

Since the optimum under valuation  $f_{A^{(0)}}$ is $1$ (achieved by $A^{(0)}$ itself), the mechanism should return expected value at least $c$. The value returned by the mechanism is exactly the random variable $X_0$, hence $\E[X_0] \geq c$.
\end{proof}

We remark that Lemma~\ref{lem:base-case} is the only place where we use the assumption of $c$-approximation.
In the following, our goal is to use the MIDR property to argue about distributions that must be in the range at higher levels, and prove successive bounds on the random variables $X_1, X_2, \ldots$.

\subsection{The symmetry gap argument}

The main building block of our proof is a symmetry gap argument whose goal is to show the following. If the mechanism optimizes over a certain fixed range $\cR$ which supports distributions of ``high density" at level $j$, then $\cR$ must support distributions of even higher density (when properly scaled) at level $j+1$. However, the way we measure density is quite intricate. Recall the random variables $X_0,X_1,\ldots,X_\ell$ that encode certain distributions in the range at each level. It would be nice to say that for any $X_j$ constructible at level $j$, there must be $X_{j+1}$ constructible at level $j+1$ such that $\E[X_{j+1}] > \frac{1+\delta}{2} \E[X_j]$ (which would correspond to sets of larger cardinality at level $j+1$ than $j$). But this is not true - the distributions of $X_j, X_{j+1}$ also matter and we cannot get a guaranteed boost just in terms of expectation. Instead, we define a measure of density using a test function $\phi$ that we specify later. 
The symmetry gap argument allows us to prove the following.

\begin{lemma}
\label{lem:sym-gap}
Let $\phi:[0,1] \rightarrow [0,1]$ be a non-decreasing concave function. Fix $j \in \{0,1,\ldots,\ell-1\}$ and 
a set $A^{(j+1)}$ of size $2^{j+1-\ell} m \geq m/n$. Let $(A^{(j)}, B^{(j)})$ be a random partition of $A^{(j+1)}$ into two sets of equal size.
Then there is a monotone submodular function $\tilde{f}_{A^{(j)}, B^{(j)}}$ for each partition $(A^{(j)}, B^{(j)})$ such that
\begin{compactitem}
\item For any distribution of a random set $R^{(j)}$ (possibly correlated with $A^{(j)}$) and the associated random variable $X_{j} = \frac{|R^{(j)} \cap A^{(j)}|}{|A^{(j)}|}$, we have
$$ \E[\tilde{f}_{A^{(j)}, B^{(j)}}(R^{(j)})] \geq \E[\phi(X_j - n m^{-1/2})]. $$
\item Any mechanism that uses $poly(n)$ value queries, when applied to the random input $\tilde{f}_{A^{(j)}, B^{(j)}}$ will return a random set $R^{(j+1)}$ such that for the random variable $X_{j+1} = \frac{|R^{(j+1)} \cap A^{(j+1)}|}{|A^{(j+1)}|}$,
$$ \E[\tilde{f}_{A^{(j)}, B^{(j)}}(R^{(j+1)})] \leq \E[1 - (1-\phi(X_{j+1}))^2] + e^{-\Omega(n)}.$$
\end{compactitem}
The expectations are over both $(A^{(j)}, B^{(j)})$ and $R^{(j)}$ or $R^{(j+1)}$ respectively.
\end{lemma}

The key point here is that the performance of a mechanism depends only on the fraction of elements taken from $A^{(j+1)}$, and not on the partition $(A^{(j)}, B^{(j)})$. While there might be a ``good distribution" $R^{(j)}$ in the range which is correlated with $A^{(j)}$, the mechanism cannot find such a distribution and must compensate for it by returning larger sets. This will be important later.

The proof of Lemma~\ref{lem:sym-gap} relies on the notion of symmetry gap developed in \cite{FMV07,MSV08,V09}. Since what we need here is a special case where the construction can be carried out explicitly quite easily, we present a self-contained proof here instead of referring to the general framework of \cite{V09}.

\begin{proof}
Consider a pair of sets  $(A^{(j)}, B^{(j)})$.
Given a non-decreasing concave function $\phi:[0,1] \rightarrow [0,1]$ ,
we define the valuation function $f_{A^{(j)},B^{(j)}}$ as follows:
$$ f_{A^{(j)},B^{(j)}}(S) = 1 - \left( 1 - \phi\left(\frac{|S \cap A^{(j)}|}{|A^{(j)}|}\right) \right)
 \left( 1 - \phi\left(\frac{|S \cap B^{(j)}|}{|B^{(j)}|}\right) \right). $$
The function depends only on how many elements we take from $A^{(j)}$ and how many from $B^{(j)}$. Moreover, the two sets play the same role in $f_{A^{(j)},B^{(j)}}$; i.e., all elements in $A^{(j+1)} = A^{(j)} \cup B^{(j)}$ contribute equivalently to $f_{A^{(j)},B^{(j)}}$. This is the kind of situation where we can apply a symmetry gap argument.

Let us simplify the notation and write $$\psi(x,y) = 1 - (1-\phi(x))(1-\phi(y)),$$ where $x,y \in [0,1]$; i.e. $f_{A^{(j)},B^{(j)}}(S) = \psi\left(\frac{|S \cap A^{(j)}|}{|A^{(j)}|}, \frac{|S \cap B^{(j)}|}{|B^{(j)}|}\right)$. It is elementary to verify that since $\phi$ is non-decreasing concave, the first partial derivatives of $\psi$ are non-negative and non-increasing with respect to both coordinates. Now we replace $\psi$ by a modified function $\tilde{\psi}$ which has the property that if $|x-y|$ is very small, the function value depends only on $x+y$. This can be accomplished explicitly as follows: For some $\beta>0$, let

\

\begin{compactitem}
\item $\tilde{\psi}(x,y) = \psi(\frac12(x+y), \frac12(x+y))$ if $|x-y| \leq \beta$.
\item $\tilde{\psi}(x,y) = \psi(x-\frac12 \beta, y+\frac12 \beta)$ if $x-y > \beta$.
\item $\tilde{\psi}(x,y) = \psi(x+\frac12 \beta, y-\frac12 \beta)$ if $y-x > \beta$.
\end{compactitem}

\

\begin{figure}[!ht]
\centering
\begin{tikzpicture}[scale=.65,pre/.style={<-,shorten <=2pt,>=stealth,thick}, post/.style={->,shorten >=1pt,>=stealth,thick}]

\draw (0,0) -- (10,0);
\draw (10,0) -- (10,10);
\draw (0,0) -- (0,10);
\draw (0,10) -- (10,10);

\draw (0,0) -- (10,10);

\draw [gray] (0,1) .. controls +(0.5,-0.5) and +(-0.5,0.5) .. (0.5,0.5);
\draw [gray] (0,2) .. controls +(0.6,-0.5) and +(-0.5,0.5) .. (1.05,1.05);
\draw [gray] (0,3) .. controls +(0.7,-0.5) and +(-0.5,0.5) .. (1.63,1.63);
\draw [gray] (0,4) .. controls +(0.8,-0.5) and +(-0.5,0.5) .. (2.25,2.25);
\draw [gray] (0,5) .. controls +(1.0,-0.5) and +(-0.5,0.5) .. (2.93,2.93);
\draw [gray] (0,6) .. controls +(1.5,-0.5) and +(-0.5,0.5) .. (3.67,3.67);
\draw [gray] (0,7) .. controls +(2.5,-0.5) and +(-0.5,0.5) .. (4.52,4.52);
\draw [gray] (0,8) .. controls +(3.5,-0.5) and +(-0.5,0.5) .. (5.53,5.53);
\draw [gray] (0,9) .. controls +(5,-0.5) and +(-0.5,0.5) .. (6.83,6.83);

\draw [gray] (1,0) .. controls +(-0.5,0.5) and +(0.5,-0.5) .. (0.5,0.5);
\draw [gray] (2,0) .. controls +(-0.5,0.6) and +(0.5,-0.5) .. (1.05,1.05);
\draw [gray] (3,0) .. controls +(-0.5,0.7) and +(0.5,-0.5) .. (1.63,1.63);
\draw [gray] (4,0) .. controls +(-0.5,0.8) and +(0.5,-0.5) .. (2.25,2.25);
\draw [gray] (5,0) .. controls +(-0.5,1.0) and +(0.5,-0.5) .. (2.93,2.93);
\draw [gray] (6,0) .. controls +(-0.5,1.5) and +(0.5,-0.5) .. (3.67,3.67);
\draw [gray] (7,0) .. controls +(-0.5,2.5) and +(0.5,-0.5) .. (4.52,4.52);
\draw [gray] (8,0) .. controls +(-0.5,3.5) and +(0.5,-0.5) .. (5.53,5.53);
\draw [gray] (9,0) .. controls +(-0.5,5) and +(0.5,-0.5) .. (6.83,6.83);

\draw (15,0) -- (25,0);
\draw (25,0) -- (25,10);
\draw (15,0) -- (15,10);
\draw (15,10) -- (25,10);

\draw [gray] (15,1) .. controls +(0.2,-0.2) and +(-0.2,0.2) .. (15.25,0.75);
\draw [gray] (15,2) .. controls +(0.3,-0.3) and +(-0.3,0.3) .. (15.80,1.30);
\draw [gray] (15,3) .. controls +(0.7,-0.5) and +(-0.5,0.5) .. (16.38,1.88);
\draw [gray] (15,4) .. controls +(0.8,-0.5) and +(-0.5,0.5) .. (17.00,2.50);
\draw [gray] (15,5) .. controls +(1.0,-0.5) and +(-0.5,0.5) .. (17.68,3.18);
\draw [gray] (15,6.05) .. controls +(1.5,-0.5) and +(-0.5,0.5) .. (18.43,3.92);
\draw [gray] (15,7.1) .. controls +(2.5,-0.5) and +(-0.5,0.5) .. (19.27,4.77);
\draw [gray] (15,8.15) .. controls +(3.5,-0.5) and +(-0.5,0.5) .. (20.28,5.78);
\draw [gray] (15,9.2) .. controls +(5,-0.5) and +(-0.5,0.5) .. (21.58,7.08);

\draw [gray] (16,0) .. controls +(-0.2,0.2) and +(0.2,-0.2) .. (15.75,0.25);
\draw [gray] (17,0) .. controls +(-0.3,0.3) and +(0.3,-0.3) .. (16.30,0.80);
\draw [gray] (18,0) .. controls +(-0.5,0.7) and +(0.5,-0.5) .. (16.88,1.38);
\draw [gray] (19,0) .. controls +(-0.5,0.8) and +(0.5,-0.5) .. (17.50,2.00);
\draw [gray] (20,0) .. controls +(-0.5,1.0) and +(0.5,-0.5) .. (18.18,2.68);
\draw [gray] (21.05,0) .. controls +(-0.5,1.5) and +(0.5,-0.5) .. (18.93,3.42);
\draw [gray] (22.1,0) .. controls +(-0.5,2.5) and +(0.5,-0.5) .. (19.77,4.27);
\draw [gray] (23.15,0) .. controls +(-0.5,3.5) and +(0.5,-0.5) .. (20.78,5.28);
\draw [gray] (24.2,0) .. controls +(-0.5,5) and +(0.5,-0.5) .. (22.08,6.58);

\draw [gray] (15.25,0.75) -- (15.75,0.25);
\draw [gray] (15.80,1.30) -- (16.30,0.80);
\draw [gray] (16.38,1.88) -- (16.88,1.38);
\draw [gray] (17.00,2.50) -- (17.50,2.00);
\draw [gray] (17.68,3.18) -- (18.18,2.68);
\draw [gray] (18.43,3.92) -- (18.93,3.42);
\draw [gray] (19.27,4.77) -- (19.77,4.27);
\draw [gray] (20.28,5.78) -- (20.78,5.28);
\draw [gray] (21.58,7.08) -- (22.08,6.58);

\draw (15.5,0) -- (25,9.5);
\draw (15,0.5) -- (24.5,10);

\draw (15.5,0) -- (25,9.5);
\draw (15,0.5) -- (24.5,10);

\draw (6.5,9) node {$\psi(x,y)$};
\draw (21.5,9) node {$\tilde{\psi}(x,y)$};

\draw (-0.2,-0.5) node {$0$};
\draw (10.2,-0.5) node {$1$};
\draw (-0.2,10.5) node {$1$};

\draw (14.8,-0.5) node {$0$};
\draw (25.2,-0.5) node {$1$};
\draw (14.8,10.5) node {$1$};

\end{tikzpicture}
\caption{Construction of $\tilde{\psi}(x,y)$ from $\psi(x,y)$, assuming that $\psi(x,y) = 1-(1-x)(1-y)$. The solid lines denote the diagonal $x=y$ and the shifted diagonals $x-y = \pm \beta$. The gray lines are the level sets of $\psi(x,y)$ and $\tilde{\psi}(x,y)$.}
\label{fig:sym-gap}
\end{figure}
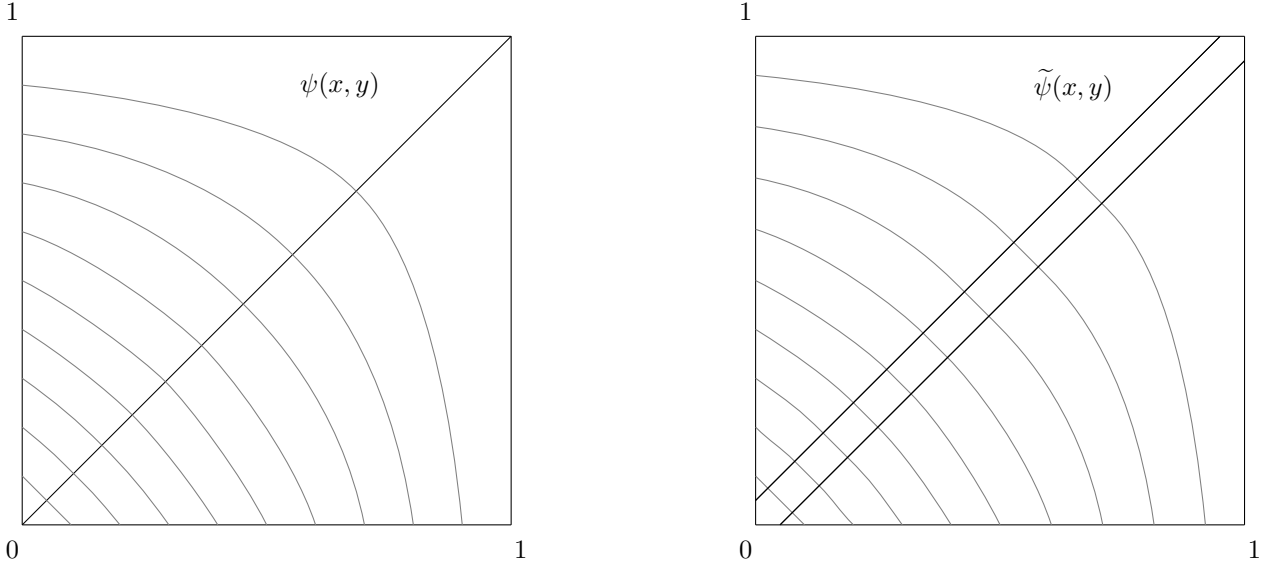

Geometrically, this construction can be seen as taking the graph of $\psi(x,y)$, pulling it away from the diagonal $x=y$ on both sides, and patching the area close to the diagonal with a function which depends only on $x+y$ and is equal to the function on the diagonal. Using the properties of $\psi$, one can check that again the first partial derivatives of $\tilde{\psi}$ are non-negative and non-increasing with respect to both coordinates. We define the function promised by the lemma as
$$\tilde{f}_{A^{(j)},B^{(j)}}(S) = \tilde{\psi}\left(\frac{|S \cap A^{(j)}|}{|A^{(j)}|},\frac{|S \cap B^{(j)}|}{|B^{(j)}|} \right).$$
The properties of $\tilde{\psi}$ imply that $\tilde{f}_{A^{(j)},B^{(j)}}$ is a monotone submodular function
(see e.g.~\cite{MSV08,V09}).

We observe the following (which is the case in all proofs using the symmetry gap). For a ``typical query" $S$, oblivious to the random partition $(A^{(j)}, B^{(j)})$, with high probability $S$ will contain approximately the same number of elements from these two sets. (Recall that $|A^{(j)}| = |B^{(j)}|$.) We call a query $S$ balanced if the parameters $x = \frac{|S \cap A^{(j)}|}{|A^{(j)}|}$ and $y = \frac{|S \cap B^{(j)}|}{|B^{(j)}|}$ are in the range where $|x-y| \leq \beta$, and hence $\tilde{f}_{A^{(j)},B^{(j)}}(S) = \tilde{\psi}(x,y) = \psi(\frac12 (x+y), \frac12(x+y))$  is independent of the particular partition $(A^{(j)}, B^{(j)})$. By Lemma~\ref{lem:bisect-chernoff} (applied to the ground set $A^{(j+1)}$), the probability that any fixed query $S$ is unbalanced is exponentially small:
 $$ \Pr[|x-y| > \beta] = \Pr\left[||S \cap A^{(j)}| - |S \cap B^{(j)}|| > \beta |A^{(j)}|\right] \leq e^{-\Omega(\beta^2 |A^{(j+1)}|)}.$$
Recall that $|A^{(j+1)}| = 2^{j+1-\ell} m \geq m/n$. Therefore, if we pick $\beta = n m^{-1/2}$, the probability is $e^{-\Omega(n)}$. 
Let us fix for now the random coin flips of the mechanism. As long as all query answers are independent of the partition $(A^{(j)}, B^{(j)})$, the mechanism will follow the same computation path, independent of $(A^{(j)}, B^{(j)})$, and we can use a union bound over its $poly(n)$ queries. Hence, the probability that a mechanism ever makes a query such that $|x-y| > \beta$ is $poly(n) e^{-\Omega(n)} = e^{-\Omega(n)}$. 
This is still true if we average over the random coin flips of the algorithm.
Therefore, the output of the mechanism will be independent of $(A^{(j)}, B^{(j)})$ with probability $1-e^{-\Omega(n)}$.

To summarize, the output of the mechanism, $R^{(j+1)}$, is with high probability independent of $(A^{(j)}, B^{(j)})$ and again by Lemma~\ref{lem:bisect-chernoff} with high probability balanced with respect to $(A^{(j)}, B^{(j)})$. Given the definition of the random variable $X_{j+1} = \frac{|R^{(j+1)} \cap A^{(j+1)}|}{|A^{(j+1)}|}$, this means the output random set contains an $X_{j+1}$-fraction of the set $A^{(j+1)}$, approximately balanced between its two halves. For some $|\beta'| \leq \frac12 \beta$, the value of such a set is
$$ \tilde{\psi}(X_{j+1}+\beta',X_{j+1}- \beta') = \psi(X_{j+1},X_{j+1}) = 1 - (1-\phi(X_{j+1}))^2.$$
Thus the expected value of this solution is  $\E[f_{A^{(j)}, B^{(j)}}(R^{(j+1)})] \leq \E[1 - (1-\phi(X_{j+1}))^2] + e^{-\Omega(n)}$ (where $e^{-\Omega(n)}$ accounts for the small probability of finding an unbalanced solution, whose value could be up to $1$). This proves the second statement of the lemma.

Finally, consider any random set $R^{(j)}$ and the associated random variable $X_{j} = \frac{|R^{(j)} \cap A^{(j)}|}{|A^{(j)}|}$. We have 
$$ \tilde{f}_{A^{(j)},B^{(j)}}(R^{(j)})) \geq  \tilde{f}_{A^{(j)},B^{(j)}}(R^{(j)} \cap A^{(j)}) = \tilde{\psi}(X_j,0)
 \geq \psi\left(X_j- \beta, 0 \right) = \phi\left(X_j - \beta \right).$$
Therefore $\E[\tilde{f}_{A^{(j)},B^{(j)}}(R^{(j)}))] \geq \E[\phi(X_j - \beta)]$. Recall that $\beta = n m^{-1/2}$, so this proves the first statement of the lemma.
\end{proof}

Considering the setup of random variables $X_0,X_1,\ldots,X_\ell$ constructible by $\cR$ at different levels
(Section~\ref{sec:basics}), we obtain the following.

\begin{lemma}
\label{lem:level-gap}
Consider a mechanism of polynomial query-complexity that $(1-\epsilon)$-approximately maximizes over a range of distributions $\cR$ for the problem $\max \{f(S): |S| \leq k\}$ for $f$ monotone submodular, ground set of size $m=400^\ell$ and $k = 2^{-\ell} m$. Let $\phi:[0,1] \rightarrow [0,1]$ be a non-decreasing concave function. If a random variable $X_j$ is constructible by $\cR$ at level $j$, then there is a random variable $X_{j+1}$ constructible by $\cR$ at level $j+1$ such that
$$\E[1 - (1-\phi(X_{j+1}))^2] \geq (1-\epsilon) \E[\phi(X_j - 10^{-\ell})] - 10^{-\ell}.$$
\end{lemma}

\begin{proof}
Given $\phi$, let $\tilde{f}_{A^{(j)}, B^{(j)}}$ be the valuation function provided by Lemma~\ref{lem:sym-gap}. Consider $A^{(j+1)}$ uniformly random among sets of size $2^{j+1-\ell} m$, bisected randomly into $A^{(j)} \cup B^{(j)}$. If $X_j$ is constructible by the range $\cR$ at level $j$, it means that for each $A^{(j)}$ there is a distribution $D(A^{(j)})$ in $\cR$ such that $X_j = \frac{|R^{(j)} \cap A^{(j)}|}{|A^{(j)}|}$ where $A^{(j)}$ is random and $X^{(j)}$ is sampled from $D(A^{(j)})$. 
By Lemma~\ref{lem:sym-gap}, conditioned on any $A^{(j+1)}$ and taking expectation over the random partition $(A^{(j)}, B^{(j)})$,
$ \E[\tilde{f}_{A^{(j)}, B^{(j)}}(R^{(j)}) \mid A^{(j+1)}] \geq \E[\phi(X_j  - n m^{-1/2}) \mid A^{(j+1)}]. $
Therefore the same holds also without the conditioning. Recall that we have $n m^{-1/2} = 2^\ell 400^{-\ell/2} = 10^{-\ell}$. So we get
$$ \E[\tilde{f}_{A^{(j)}, B^{(j)}}(R^{(j)})] \geq \E[\phi(X_j  - n m^{-1/2})] =  \E[\phi(X_j  - 10^{-\ell})].$$

Now let us run the mechanism on the same random instance and denote the output random set by $R^{(j+1)}$.
By Lemma~\ref{lem:sym-gap}, 
$ \E[\tilde{f}_{A^{(j)}, B^{(j)}}(R^{(j+1)}) \mid A^{(j+1)}]
 \leq \E[1 - (1-\phi(X_{j+1}))^2 \mid A^{(j+1)}] + e^{-\Omega(n)}$, where $X_{j+1} = \frac{|R^{(j+1)} \cap A^{(j+1)}|}{|A^{(j)}|}$.
Hence this holds also without the conditioning: 
$$ \E[\tilde{f}_{A^{(j)}, B^{(j)}}(R^{(j+1)})] \leq \E[1 - (1-\phi(X_{j+1}))^2] + e^{-\Omega(n)}
 \leq \E[1 - (1-\phi(X_{j+1}))^2] + 10^{-\ell} $$
 and by definition $X_{j+1}$ is constructible by $\cR$ at level $j+1$.

To conclude, if the mechanism maximizes $(1-\epsilon)$-approximately over $\cR$, then the expected value of $R^{(j+1)}$ conditioned on $(A^{(j)}, B^{(j)})$ must be at least $(1-\epsilon) \times$ that provided by $R^{(j)}$. Therefore, the same holds in expectation over $(A^{(j)}, B^{(j)})$, which means $\E[\tilde{f}_{A^{(j)}, B^{(j)}}(R^{(j+1)})] \geq (1-\epsilon) \E[\tilde{f}_{A^{(j)}, B^{(j)}}(R^{(j)})]$ and the lemma follows.
\end{proof}

\subsection{The gap amplification argument}

In this section, we develop an inductive argument based on Lemma~\ref{lem:base-case} and Lemma~\ref{lem:level-gap}, which proves that a certain notion of density of the distributions at level $j$ increases exponentially in $j$. 
By Lemma~\ref{lem:level-gap}, for any $X_j$ constructible at level $j$ there is $X_{j+1}$ constructible at level $j+1$ such that
$$ \E[1 - (1-\phi(X_{j+1}))^2] \geq (1-\epsilon) \E[\phi(X_j - 10^{-\ell})] - 10^{-\ell}.$$
We 
want to prove that $X_{j+1}$ is in some sense ``significantly larger" than $\frac12 X_j$.
Our main technical lemma formalizing this intuition is the following.

\begin{lemma}
\label{lem:level-bound}
There are absolute constants $\epsilon,\delta > 0$ such that the following holds for any sufficiently large $\ell \in \NN$. If $\cX_0,\ldots,\cX_\ell$ are collections of random variables in $[0,1]$ such that
\begin{itemize}
\item there is $X_0$ in $\cX_0$ such that $\E[X_0] \geq c$ for some $c \geq 2^{-\ell}$, and
\item for every $X_j$ in $\cX_j$ and every non-decreasing concave function $\phi:[0,1] \rightarrow [0,1]$,
there is $X_{j+1}$ in $\cX_{j+1}$ such that
$$ \E[1 - (1-\phi(X_{j+1}))^2] \geq (1-\epsilon) \E[\phi(X_j - 10^{-\ell})] - 10^{-\ell}$$
\end{itemize}
then there is a sequence of variables $X_j$ in $\cX_j$ and parameters $1 = \alpha_0 \geq \alpha_1 \geq \ldots \alpha_\ell > 0$
such that if we define
$\phi_\alpha(t) = \min \left\{ \frac{t}{\alpha}, 1 \right\} $ then
$$ \alpha_j (\E[\phi_{\alpha_j}(X_j)])^{1+\delta} \geq \left( \frac{1+\delta^2}{2} \right)^j c^{1+\delta}.$$
\end{lemma}

The use of $1+\delta$ in the exponent is crucial here; note that it makes the statement stronger, but this is what makes the inductive proof work.
The intuitive meaning of this lemma is as follows: there exist random variables $X_j$ constructible at different levels that, when measured by suitable test functions, decrease roughly as $\left(\frac{1+\delta^2}{2}\right)^j$, rather than $\frac{1}{2^j}$. In terms of the cardinality of the returned sets, this means they increase by a factor of $(1+\delta^2)$ at each level. This gives the exponential amplification that we need.

\begin{proof}
The base case $j=0$ holds trivially with $\alpha_0 = 1$ and $\phi_{\alpha_0}(t) = t$.
To prove the inductive step, suppose that there is $\alpha_j \in [0,1]$ that satisfies the statement of the lemma for $X_j$. Let us define $\xi_j = \E[\phi_{\alpha_j}(X_j)]$; then the inductive statement reads
\begin{equation}
\label{eq:hypo}
\alpha_j \xi_j^{1+\delta} \geq \left( \frac{1+\delta^2}{2} \right)^j c^{1+\delta}.
\end{equation}
Our goal is to prove that $\alpha_{j+1} \xi_{j+1}^{1+\delta} \geq \frac{1+\delta^2}{2} \alpha_j \xi_j^{1+\delta}$, which implies the inductive statement for $j+1$.

By assumption, for the non-decreasing concave function $\phi_{\alpha_j}$, we get
$$ \E[1 - (1-\phi_{\alpha_j}(X_{j+1}))^2] \geq (1-\epsilon) \E[\phi_{\alpha_j}(X_j - 10^{-\ell})] - 10^{-\ell}.$$
First, we simplify the error terms on the right-hand side. Let us keep in mind that $\epsilon,\delta>0$ are (small) absolute constants which will be suitably chosen at the end of the proof. 
Recall that $\phi_{\alpha_j}(t) = \min \{ \frac{t}{\alpha_j},1\}$. Therefore,
 $(1-\epsilon) \E[\phi_{\alpha_j}(X_j - 10^{-\ell})] \geq (1-\epsilon) \E[\phi_{\alpha_j}(X_j)] - \frac{1}{\alpha_j 10^\ell}
  = (1-\epsilon) \xi_j - \frac{1}{\alpha_j 10^\ell}.$
Recall the inductive hypothesis (\ref{eq:hypo}).
Since $\alpha_j, \xi_j \in [0,1]$, and $c \geq 2^{-\ell}$, this means in particular that $\alpha_j \xi_j \geq  2^{-j} c^{1+\delta} \geq 2^{-3 \ell}$. Also, $\xi_j \geq 2^{-j} c \geq 2^{-2 \ell}$. Hence, we can estimate 
\begin{equation}
\label{eq:quad}
\E[1 - (1-\phi_{\alpha_j}(X_{j+1}))^2] \geq (1-\epsilon) \xi_j - \frac{1}{\alpha_j 10^\ell} - \frac{1}{10^\ell} 
\geq \left(1-\epsilon - \frac{2^{3\ell}}{10^\ell} - \frac{2^{2\ell}}{10^\ell} \right) \xi_j
 \geq (1-2\epsilon) \xi_j
\end{equation}
for $\ell$ sufficiently large.
 
Now we come to the meat of the inductive argument. Instead of the expression $\E[1 - (1-\phi_{\alpha_j}(X_{j+1}))^2]$, we would like to estimate $\xi_{j+1} = \E[\phi_{\alpha_{j+1}}(X_{j+1})]$ for a suitable value of $\alpha_{j+1}$. The reason why the values of $\alpha_j$ are not specified by the lemma is that their choice depends on the particular distributions of $X_j$ over which we have no control. For example, $\alpha_{j+1} = \frac12 \alpha_j$ is a natural choice which works for some distributions of $X_{j+1}$ but not always. In the following, we split the analysis into 2 cases.
 
\paragraph{Case 1:} $\Pr[\frac{X_{j+1}}{\alpha_j} > \sqrt{\delta}] > 2 \delta \xi_j$. \\
In this case, $X_{j+1}$ is with non-negligible probability quite large, in the region where $1 - (1-X_{j+1}/\alpha_j)^2$ is significantly smaller than $2 X_{j+1} / \alpha_j$ . In this case, we can gain by making $\alpha_{j+1}$ slightly larger than $\frac12 \alpha_j$, specifically $\alpha_{j+1} = \frac12 (1+\delta) \alpha_j$.
We obtain:
\begin{eqnarray*}
\xi_{j+1} = \E[\phi_{\alpha_{j+1}}(X_{j+1})] = \E\left[ \min \left\{ \frac{X_{j+1}}{\alpha_{j+1}}, 1 \right\}\right]
 = \E\left[ \min \left\{ \frac{2 X_{j+1}}{(1+\delta) \alpha_j}, 1 \right\} \right] 
  = \frac{1}{1+\delta} \E\left[ \min \left\{ \frac{2 X_{j+1}}{\alpha_j}, 1+\delta \right\}\right].
\end{eqnarray*}
Observe the following: $\min \{ \frac{2 X_{j+1}}{\alpha_j}, 1+\delta \} \geq \min \{ \frac{2 X_{j+1}}{\alpha_j}, 1 \}
 \geq 1 - (1-\phi_{\alpha_j}(X_{j+1}))^2$ for all $X_{j+1} \geq 0$. Moreover, 
if $X_{j+1} > \sqrt{\delta} \alpha_j$, we gain an additional $\delta$, because then
$$ \min \left\{ \frac{2 X_{j+1}}{\alpha_j}, 1+\delta \right\}
 \geq 1 - \left(1-\min \left\{ \frac{X_{j+1}}{\alpha_j}, 1 \right\} \right)^2 + \delta
 = 1 - (1-\phi_{\alpha_j}(X_{j+1}))^2 + \delta $$
(with equality for $X_{j+1} = \sqrt{\delta} \alpha_j$ and $X_{j+1} \geq \alpha_j$; the best way to verify this is to ponder the graph
in Figure~\ref{fig:phi-graph}).
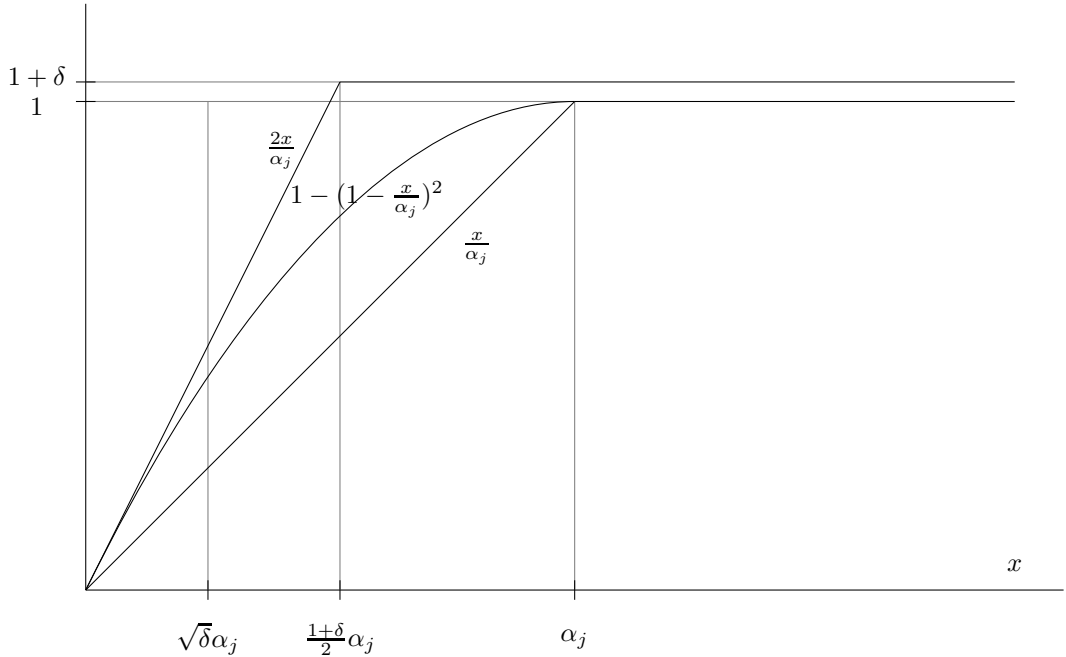
\begin{figure}[!ht]
\centering
\begin{tikzpicture}[scale=.65,pre/.style={<-,shorten <=2pt,>=stealth,thick}, post/.style={->,shorten >=1pt,>=stealth,thick}]

\draw (0,0) -- (20,0);
\draw (0,0) -- (0,12);

\draw [gray] (0,10) -- (10,10);
\draw [gray] (0,10.4) -- (5.2,10.4);
\draw (-0.2,10) -- (0.2,10);
\draw (-0.2,10.4) -- (0.2,10.4);
\draw (-1,9.9) node {$1$};
\draw (-1,10.5) node {$1+\delta$};

\draw [gray] (2.5,0) -- (2.5,10);
\draw [gray] (5.2,0) -- (5.2,10.4);
\draw [gray] (10,0) -- (10,10);
\draw (2.5,-0.2) -- (2.5,0.2);
\draw (5.2,-0.2) -- (5.2,0.2);
\draw (10,-0.2) -- (10,0.2);
\draw (2.5,-1) node {$\sqrt{\delta} \alpha_j$};
\draw (5.2,-1) node {$\frac{1+\delta}{2} \alpha_j$};
\draw (10,-1) node {$\alpha_j$};

\draw (0,0) -- (10,10);
\draw (10,10) -- (19,10);
\draw (8,7) node {$\frac{x}{\alpha_j}$};

\draw (0,0) .. controls (1,2) and (5,10) .. (10,10);
\draw (5.75,8) node {$1 - (1-\frac{x}{\alpha_j})^2$};

\draw (0,0) -- (5.2,10.4);
\draw (5.2,10.4) -- (19,10.4);
\draw (4,9) node {$\frac{2 x}{\alpha_j}$};

\draw (19,0.5) node {$x$};

\end{tikzpicture}
\caption{Comparison of the 3 relevant functions for Case 1: Note that for $x \geq \sqrt{\delta} \alpha_j$, the top two functions differ by at least $\delta$; i.e, $\min \{\frac{2x}{\alpha_j}, 1+\delta\} \geq 1 - (1 - \min \{\frac{x}{\alpha_j}, 1\})^2 + \delta$.}
\label{fig:phi-graph}
\end{figure}
Therefore,
\begin{eqnarray*}
\xi_{j+1} = \frac{1}{1+\delta} \E\left[ \min \left\{ \frac{2 X_{j+1}}{\alpha_j}, 1+\delta \right\} \right]
 \geq \frac{1}{1+\delta}\E[1 - (1-\phi_{\alpha_j}(X_{j+1}))^2] + \frac{\delta}{1+\delta} \Pr[X_{j+1} > \sqrt{\delta} \alpha_j].
\end{eqnarray*}
Using (\ref{eq:quad}) and $\Pr[X_{j+1} > \sqrt{\delta} \alpha_j] > 2 \delta \xi_j$, we get
$$ \xi_{j+1} \geq 
 \frac{1-2\epsilon}{1+\delta} \xi_j + \frac{2 \delta^2}{1+\delta} \xi_j =
\frac{1+2\delta^2-2\epsilon}{1+\delta} \xi_j. $$
Since $\alpha_{j+1} = \frac{1+\delta}{2} \alpha_j$, we get
$$ \alpha_{j+1} \xi_{j+1}^{1+\delta}
\geq \frac{1+\delta}{2} \alpha_j \left( \frac{1+2\delta^2-2\epsilon}{1+\delta}\right)^{1+\delta} \xi_j^{1+\delta} 
= \frac{(1+2\delta^2-2\epsilon)^{1+\delta}}{2(1+\delta)^\delta} \alpha_j \xi_j^{1+\delta}.$$
We choose $\epsilon = \delta^4$, so that $(1+2\delta^2-2\epsilon)^{1+\delta} = (1+2\delta^2-2\delta^4)^{1+\delta}\geq 1 + 2\delta^2 + \delta^4$ (it can be verified that this holds for $\delta \in [0,\frac12]$). We also use $(1+\delta)^\delta \leq 1+\delta^2$ (which holds for $\delta \in [0,1]$). This implies the inductive statement:
$$ \alpha_{j+1} \xi_{j+1}^{1+\delta} \geq \frac{1+2\delta^2+\delta^4}{2(1+\delta^2)} \alpha_j \xi_j^{1+\delta}
= \frac{1+\delta^2}{2} \alpha_j \xi_j^{1+\delta}.$$ 

\paragraph{Case 2:} $\Pr[\frac{X_{j+1}}{\alpha_j} > \sqrt{\delta}] \leq 2 \delta \xi_j$. \\
In this case, $X_{j+1}$ is almost always very small compared to $\alpha_j$. Then we can gain by making $\alpha_{j+1}$ much smaller than $\alpha_j$; we let $\alpha_{j+1} = \sqrt{\delta} \alpha_j$. We have
\begin{eqnarray*}
\xi_{j+1} = \E[\phi_{\alpha_{j+1}}(X_{j+1})] & = & \E\left[ \min \left\{ \frac{X_{j+1}}{\alpha_j \sqrt{\delta}}, 1 \right\} \right]
 = \frac{1}{\sqrt{\delta}} \, \E\left[ \min \left\{ \frac{X_{j+1}}{\alpha_j}, \sqrt{\delta} \right\} \right] \\
 & \geq & \frac{1}{\sqrt{\delta}} \left( \E\left[ \min \left\{ \frac{X_{j+1}}{\alpha_j}, 1 \right\}\right] - (1-\sqrt{\delta}) \Pr\left[\frac{X_{j+1}}{\alpha_j} > \sqrt{\delta}\right] \right) \\
& \geq & \frac{1}{\sqrt{\delta}} \left( \E[\phi_{\alpha_j}(X_{j+1})] - (1-\sqrt{\delta}) \cdot 2 \delta \xi_j \right)
 \end{eqnarray*}
An elementary bound together with (\ref{eq:quad}) gives
$$ \E[\phi_{\alpha_j}(X_{j+1})] \geq \frac12 \E[1 - (1-\phi_{\alpha_j}(X_{j+1}))^2]
 \geq \frac12 (1-2\epsilon) \xi_j.$$
Therefore, using our choice of $\epsilon = \delta^4$,
$$ \xi_{j+1} = \E[\phi_{\alpha_{j+1}}(X_{j+1})]
 \geq \frac{1}{\sqrt{\delta}} \left(\frac12 (1-2\epsilon) \xi_j - 2 (1-\sqrt{\delta}) \delta \xi_j\right)
  = \frac{1 - 2\delta^4 - 4 \delta + 4\delta^{3/2}}{2 \sqrt{\delta}} \xi_j \geq \frac{1-4\delta+2 \delta^{3/2}}{2\sqrt{\delta}} \xi_j. $$
From here, using $\alpha_{j+1} = \sqrt{\delta} \alpha_j$ and $(1-4\delta+2\delta^{3/2})^{1+\delta}
 \geq 1 - 4\delta$ (which holds for any $\delta \in [0,\frac14]$),
$$ \alpha_{j+1} \xi_{j+1}^{1+\delta} \geq \sqrt{\delta} \alpha_j \left(\frac{1-4\delta+2\delta^{3/2}}{2\sqrt{\delta}}\right)^{1+\delta} \xi_j^{1+\delta} \geq \frac{1-4\delta}{2^{1+\delta} \delta^{\delta/2}} \alpha_j \xi_j^{1+\delta}.$$
We choose $\delta = e^{-10}$ so that $\delta^{\delta/2} = e^{-5\delta}$. Then,
$$ \alpha_{j+1} \xi_{j+1}^{1+\delta} \geq \frac{1-4\delta}{2^{1+\delta}} e^{5\delta} \alpha_j \xi_j^{1+\delta}
 \geq \frac{1+\delta^2}{2} \alpha_j \xi_j^{1+\delta} $$
which finishes the inductive step.
\end{proof}

Putting together Lemma~\ref{lem:level-bound} and the cardinality bound which applies to every feasible solution, we complete our hardness result for combinatorial public projects.

\begin{proof}[Proof of Theorem~\ref{thm:CPP-hardness}]
Let $\epsilon>0$ and $\delta>0$ be the constants provided by Lemma~\ref{lem:level-bound}. Let $n=2^\ell$ and $m = 400^\ell$. Suppose there is a mechanism for the problem $\max \{f(S): |S| \leq m/n\}$ that maximizes $(1-\epsilon)$-approximately over a distributional range $\cR$ and provides a $c$-approximation, where $c \geq 1/n$. By Lemma~\ref{lem:level-gap} and Lemma~\ref{lem:base-case}, there are collections of random variables $\cX_0, \cX_1, \ldots, \cX_\ell$ constructible at the respective levels by $\cR$, satisfying the conditions of Lemma~\ref{lem:level-bound}. Hence,
by Lemma~\ref{lem:level-bound} for $j=\ell$, there is $X_\ell$ constructible by $\cR$ at level $\ell$ such that
$$ \alpha_\ell (\E[\phi_{\alpha_\ell}(X_\ell)])^{1+\delta} \geq c^{1+\delta} \left( \frac{1+\delta^2}{2} \right)^\ell  = \frac{c^{1+\delta}}{n} (1+\delta^2)^\ell.$$
Recall that $\phi_{\alpha_\ell}(t) = \min \{ \frac{t}{\alpha_\ell}, 1\}$. Therefore, we have 
$$ \frac{c^{1+\delta}}{n} (1+\delta^2)^\ell \leq \alpha_\ell (\E[\phi_{\alpha_\ell}(X_\ell)])^{1+\delta} \leq 
\alpha_\ell \E[\phi_{\alpha_\ell}(X_\ell)] \leq \E[X_\ell]. $$
We have $X_\ell = \frac{|R|}{|M|}$ where $R$ is a random set sampled according to some distribution in the range $\cR$. 
All distributions in the range must be feasible in expectation, otherwise the mechanism cannot possibly maximize over them and return a feasible solution. Therefore, 
$ \E[X_\ell] \leq \frac{1}{n} $
which implies that
$$ c \leq (1+\delta^2)^{-\frac{\ell}{1+\delta}} < 2^{-\delta^2 \ell} = n^{-\delta^2}. $$
Therefore, there is no $(1-\epsilon)$-approximately MIDR mechanism providing an $n^{-\delta^2}$-approximation in the objective function.
Also, we have $m = 400^\ell = \poly(n)$, so the approximation cannot be better than $m^{-\gamma}$ for some constant $\gamma>0$.
The only bound we have used on the mechanism was that the number of value queries is polynomial in $n$, or equivalently
polynomial in $m$.
\end{proof}

\section{Proof of hardness for combinatorial auctions}
\label{sec:CA-proof}

In this section, we present the proof of Theorem~\ref{thm:CA-hardness}.

\subsection{The basic random instance}
\label{sec:basic-inst}

We choose a parameter $\ell \geq 1$ and construct instances with $|N| = n = 2^\ell$ players and $|M| = m = 400^\ell$ items. We define "polar valuations" as in \cite{Dobzin11}.

\begin{definition}
Given a set of items $A \subset M$ and a parameter $\omega>0$, the polar valuation $v^*_A$ associated with $A$ is defined by 
$$ v^*_A(S) = |A \cap S| + \omega |S \setminus A|.$$
\end{definition}

Our "basic instance" is an instance where each player has a polar valuation associated with a random set of size $m/n$.

\begin{definition}
In the basic instance, player $i$ has valuation $v^*_i = v^*_{A^{(0)}_i}$ where $A^{(0)}_i$ is a uniformly random set of size $m/n$, chosen independently for each player.
\end{definition}

Next, we prove that for some player, his allocation overlaps significantly with his desired set.

\begin{lemma}
\label{lem:basic-inst}
For any $c$-approximation mechanism applied to the random basic instance, there is a player $i$ and sets $A^{(0)}_j, j \neq i$, such that conditioned on the desired sets for players $j \neq i$ being $A^{(0)}_j$, player $i$ gets allocated a random set $R^{(0)}_i$ such that
$$ \E[|R^{(0)}_i \cap A^{(0)}_i|] > (c/4-\omega) \E[|R^{(0)}_i \cup A^{(0)}_i|].$$
\end{lemma}

\begin{proof}
First, let us estimate the optimal social welfare that the basic instance admits in expectation. Given $(A^{(0)}_1,\ldots,A^{(0)}_n)$, each item in $\bigcup_{i=1}^{n} A^{(0)}_i$ can be allocated to some player so that it brings value $1$. We ignore the remaining items. Observe that a fixed item $j$ appears in each $A^{(0)}_i$ independently with probability $1/n$, therefore $Pr[j \in \bigcup_{i=1}^{n} A^{(0)}_i] = 1 - (1-1/n)^n \geq 1-1/e > 1/2$. Hence,
$$ \E[OPT] \geq \E[|\bigcup_{i=1}^{n} A^{(0)}_i|] > \frac{m}{2}.$$
We remind the reader that the expectation is over the random choices of $(A^{(0)}_1,\ldots,A^{(0)}_n)$. A $c$-approximate mechanism should provide at least $c \cdot OPT$ in expectation for every particular instance. Hence also in expectation over the random choice of $(A^{(0)}_1,\ldots,A^{(0)}_n)$. If $(R_1,\ldots,R_n)$ is the allocation provided by the mechanism, this means
$$ \sum_{i=1}^{n} \E[|R_i \cap A^{(0)}_i| + \omega|R_i \setminus A^{(0)}_i|] \geq c \cdot OPT > \frac{c m}{2}.$$
Since each of $A^{(0)}_1,\ldots,A^{(0)}_n$ has size $m/n$ and the sizes of $R_1,\ldots,R_n$ add up to at most $m$,
we can write
$$ \sum_{i=1}^{n} \E[|R_i \cap A^{(0)}_i| + \omega|R_i \setminus A^{(0)}_i|] > \frac{cm}{2} \geq \frac{c}{4} \sum_{i=1}^{n} \E[|R_i| + |A^{(0)}_i|] \geq \frac{c}{4} \sum_{i=1}^{n} \E[|R_i \cup A^{(0)}_i|].$$
By an averaging argument, there must be $i$ such that
$$ \E[|R_i \cap A^{(0)}_i| + \omega|R_i \setminus A^{(0)}_i|] > \frac{c}{4} \E[|R_i \cup A^{(0)}_i|] $$
and therefore
$$ \E[|R_i \cap A^{(0)}_i|] > (c/4 - \omega) \E[|R_i \cup A^{(0)}_i|].$$
This holds in expectation over the choices of $(A^{(0)}_j: j \neq i)$, and again by an averaging argument it also holds conditioned on some particular choice of $(A^{(0)}_j: j \neq i)$. We call the random set allocated to player $i$ under this conditioning $R^{(0)}_i$.
\end{proof}

\subsection{Setup for higher-level valuations}
\label{sec:prep}

In the following, the valuations of all players except $i$ are fixed to be $v^*_j = v^*_{A^{(0)}_j}$ for some choice of sets $(A^{(0)}_j: j \neq i)$. Now we will vary the valuation of player $i$ in order to be able to apply the symmetry gap argument as before. Since we work only with player $i$, we call him the "special player" and we drop the index $i$ in the following.

Recall Definition~\ref{def:rnd-bisect}, the definition of a random bisection sequence. We will use the same concept here, where at level $j$ we have a random set $A^{(j)}$ partitioned randomly into $A^{(j-1)} \cup B^{(j-1)}$. These sets have sizes $|A^{(j-1)}| = |B^{(j-1)}| = 2^{j-1-\ell} m$.
We will use valuation functions associated with each pair of sets. We denote these valuation functions by $v_{A^{(j-1)},B^{(j-1)}}$ at level $j$. In particular, this valuation function depends only on the elements of $A^{(j)} = A^{(j-1)} \cup B^{(j-1)}$. In other words, $A^{(j-1)}$ and $B^{(j-1)}$ are the desired sets of items at level $j$. 

\paragraph{Distribution menu.}
For each particular choice of a valuation function $v_{A^{(j)},B^{(j)}}$ at a certain level, the mechanism needs to produce a distribution over item sets for the special player, along with a certain price. Recall that due to the definition of (approximate) truthfulness in expectation, this choice should give (approximately) the optimal utility for the special player among all possible choices given the other valuations $v^*_{-i}$. After fixing a set of valuation functions $v_{A^{(j)},B^{(j)}}$ for each pair $(A^{(j)},B^{(j)})$, the output distribution will depend only on $(A^{(j)},B^{(j)})$ and hence we denote the respective random set by $R(A^{(j)},B^{(j)})$; we also denote the associated price by $P(A^{(j)},B^{(j)})$. Thus the mechanisms assigns distributions over sets and prices to all pairs of sets. As before, we distill the important information from the distribution into a random variable $X_j$. There is some additional information now expressed by the price; we associate the price with a separate random variable $P_j$. The possible choices of distributions for $(X_j,P_j)$ are what we call a {\em distribution menu} at level $j$.

\begin{definition}
\label{def:menu}
Given a mechanism and a special player with other valuations fixed, the "distribution menu at level $j$", $\cM_j$, is the set of all probability distributions of a pair of variables $(X_j,P_j)$ that arise as follows: There exist valuations $v_{A^{(j-1)},B^{(j-1)}}$ such that when declaring  $v_{A^{(j-1)},B^{(j-1)}}$, the special player receives a random set $R(A^{(j-1)},B^{(j-1)})$ at a price $P(A^{(j-1)},B^{(j-1)})$. Then, for $A^{(j)} = A^{(j-1)} \cup B^{(j-1)}$ chosen as the $(j-1)$-th level of a random bisection sequence, i.e. a random pair of disjoint sets of size $2^{j-1-\ell} m$, we have
$$ X_j = \frac{|A^{(j)} \cap R(A^{(j-1)},B^{(j-1)})|}{|A^{(j)}|}, $$
$$ P_j = P(A^{(j-1)},B^{(j-1)}). $$
\end{definition}

In other words, $X_j$ encodes the (random) fraction of the relevant items that the special player receives at level $j$, and $P_j$ is the respective (random) price.
Note that there are two sources of randomness in $(X_j,P_j)$: one is the random choice of $(A^{(j-1)},B^{(j-1)})$, and one arises from the randomness of the mechanism for fixed $(A^{(j-1)},B^{(j-1)})$.

\paragraph{Closure of a distribution menu.}
Furthermore, it will be convenient to make the menu closed and convex as follows.

\begin{definition}
\label{def:menu-hull}
We define $\overline{\cM_j}$, the closure of the distribution menu at level $j$, to be the topological closure of the set of all convex combinations of distributions from the menu $\cM_j$. 
\end{definition}

I.e., we take the convex hull of the menu and then its topological closure. 
We emphasize that the convex hull is generated by averaging distributions, and not the values of $(X_j,P_j)$. In other words, a  distribution of $(X'_j,P'_j)$ is in $\overline{\cM_j}$ if its distribution can be approximated arbitrarily closely by some convex combination of distributions in $\cM_j$. It is important that we keep all the randomness present in $(X_j,P_j)$ and do not take expectations until the end.

\subsection{Symmetry gap revisited}

Recall Lemma~\ref{lem:sym-gap} which was proved using the symmetry gap argument and played an important role in our proof for the CPP problem.  We still want to use this lemma; however, the difference now is the presence of prices.  In order to deal with prices, we need to introduce a parameter $\lambda$ which acts as a conversion factor between values and prices. For that purpose, we prove the following slight variation of Lemma~\ref{lem:sym-gap}.

\begin{lemma}
\label{lem:sym-gap-2}
Let $\phi:[0,1] \rightarrow [0,1]$ be a non-decreasing concave function and $\lambda \geq 0$ any constant.
Let $A^{(j+1)}$ be a fixed set of size $2^{j+1-\ell} m \geq m/n$ and $(A^{(j)}, B^{(j)})$ a random partition of $A^{(j+1)}$ into two sets of equal size.
Then there is a monotone submodular function $\tilde{v}_{A^{(j)}, B^{(j)}}$ for each partition $(A^{(j)}, B^{(j)})$ such that
\begin{compactitem}
\item For any distribution of a random set $R^{(j)}$ (possibly correlated with $A^{(j)}$) and the associated random variable $X_{j} = \frac{|R^{(j)} \cap A^{(j)}|}{|A^{(j)}|}$, we have
$$ \E[\tilde{v}_{A^{(j)}, B^{(j)}}(R^{(j)})] \geq \lambda \E[\phi(X_j - n m^{-1/2})]. $$
\item Any mechanism that uses $poly(n)$ value queries, when applied to the random input $\tilde{v}_{A^{(j)}, B^{(j)}}$ will return a random set $R^{(j+1)}$ such that for the random variable $X_{j+1} = \frac{|R^{(j+1)} \cap A^{(j+1)}|}{|A^{(j+1)}|}$,
$$ \E[\tilde{v}_{A^{(j)}, B^{(j)}}(R^{(j+1)})] \leq \lambda \E[1 - (1-\phi(X_{j+1}))^2 + e^{-\Omega(n)}].$$
\end{compactitem}
The expectations are over both $(A^{(j)}, B^{(j)})$ and $R^{(j)}$ or $R^{(j+1)}$ respectively.
\end{lemma}

\begin{proof}
The proof is easily obtained from the proof of Lemma~\ref{lem:sym-gap}. The only difference is the scaling by $\lambda \geq 0$. (For $\lambda=0$ the statement is trivial.) Given $\phi:[0,1] \rightarrow [0,1]$ and $\lambda \geq 0$, we take the function $\tilde{f}_{A^{(j)}, B^{(j)}}$ provided by Lemma~\ref{lem:sym-gap} and scale it by $\lambda$:
$$ \tilde{v}_{A^{(j)}, B^{(j)}}(S) = \lambda \tilde{f}_{A^{(j)}, B^{(j)}}(S).$$
Randomizing over $(A^{(j)}, B^{(j)})$, the same proof shows that any mechanism will return a random set $R^{(j+1)}$ with high probability balanced with respect to $(A^{(j)}, B^{(j)})$. Hence we obtain the same bounds as in Lemma~\ref{lem:sym-gap} with the right-hand side scaled by $\lambda$.
\end{proof}

Applying the assumption of approximate truthfulness, we obtain the following.

\begin{lemma}
\label{lem:lambda-gap}
Consider a $(1-\epsilon)$-approximately truthful-in-expectation mechanism for combinatorial auctions with $n=2^\ell$ players and $m=400^\ell$ items. Let $\phi:[0,1] \rightarrow [0,1]$ be a non-decreasing concave function. If $(X_j,P_j)$ has a distribution in the closure of the level-$j$ menu $\overline{\cM_j}$, then for any $\lambda',\lambda'' \geq 0$ there is $(X_{j+1},P_{j+1})$ with a distribution in the closure of the level-$(j+1)$ menu $\overline{\cM_{j+1}}$ such that 
$$ \lambda' \E[1 - (1-\phi(X_{j+1}))^2] - \lambda'' \E[P_{j+1}] \geq \lambda' \E[(1-\epsilon) \phi(X_j-10^{-\ell}) - 10^{-\ell}] - \lambda'' \E[P_j].$$
\end{lemma}

\begin{proof}
First let us assume that $\lambda''>0$ and set $\lambda = \lambda'/\lambda''$. If $(X_j,P_j)$ is on the menu $\cM_j$, it means that the mechanism under certain valuations depending on the (random) set $A^{(j)}$ allocates to the special player a random set $R^{(j)}$ (at some price $P_j$) such that $X_{j} = \frac{|R^{(j)} \cap A^{(j)}|}{|A^{(j)}|}$. Given $\phi$, by Lemma~\ref{lem:sym-gap-2} there are valuation functions $\tilde{v}_{A^{(j)}, B^{(j)}}$ such that 
$$ \E[\tilde{v}_{A^{(j)}, B^{(j)}}(R^{(j)})] \geq \lambda \E[\phi(X_j - n m^{-1/2})] 
$$
and on the other hand, the mechanism executed on this random input allocates a random set $R^{(j+1)}$ such that with $X_{j+1} = \frac{|R^{(j+1)} \cap A^{(j+1)}|}{|A^{(j+1)}|}$,
$$ \E[\tilde{v}_{A^{(j)}, B^{(j)}}(R^{(j+1)})] \leq \lambda \E[1 - (1-\phi(X_{j+1}))^2 + e^{-\Omega(n)}]
.$$
Let us assume that the mechanism allocates this distribution at a price $P_{j+1}$. 
Due to the assumption of $(1-\epsilon)$-truthfulness, the utility provided by the mechanism must be approximately maximized for the true valuation. Hence, we must have
\begin{equation}
\label{eq:lambda-bound}
 \lambda \E[1 - (1-\phi(X_{j+1}))^2 + e^{-\Omega(n)}] - \E[P_{j+1}] \geq (1-\epsilon) \lambda \E[\phi(X_j-n m^{-1/2})] - \E[P_j].
\end{equation}
Given our parameters $m=400^\ell, n=2^\ell$, we have $n m^{-1/2} = 10^{-\ell}$ and $e^{-\Omega(n)} = e^{-\Omega(2^\ell)} << 10^{-\ell}$, therefore
$$ \lambda \E[1 - (1-\phi(X_{j+1}))^2] - \E[P_{j+1}] \geq \lambda \E[(1-\epsilon) \phi(X_j - 10^{-\ell}) - 10^{-\ell}] - \E[P_j].$$
Since this inequality is preserved under convex combinations and limits of the distributions of $(X_j,P_j)$ and $(X_{j+1},P_{j+1})$ (the non-linearity of $\phi$ is irrelevant here!), the same holds for the closures $\overline{\cM_j}, \overline{\cM_{j+1}}$: For any $(X_j,P_j)$ with a distribution in $\overline{\cM_j}$ and $\lambda>0$, there exists $(X_{j+1},P_{j+1})$ with a distribution in $\overline{\cM_{j+1}}$ such that $(\ref{eq:lambda-bound})$ holds. This proves the statement of the lemma when $\lambda''>0$.

When $\lambda''=0$, the statement claims that given $(X_j,P_j) \in \overline{\cM_j}$, there is $(X_{j+1},P_{j+1}) \in \overline{\cM_{j+1}}$, such that $\E[1 - (1-\phi(X_{j+1}))^2] \geq \E[(1-\epsilon) \phi(X_j-10^{-\ell}) - 10^{-\ell}]$,
without regard to prices. This can be obtained from the previous discussion as follows. Let us assume that $p^*$  is an absolute lower bound on the expected price $\E[P_{j+1}]$. (If the mechanism possibly pays arbitrarily large amounts on the menu of the special player, then given a zero valuation it cannot maximize utility over the menu.)
Given $(X_j,P_j)$ in $\overline{\cM_j}$, let $\lambda = 10^{\ell+1} (\E[P_j] - p^*)$; there must be a pair $(X_{j+1},P_{j+1})$ in $\overline{\cM_{j+1}}$ satisfying
 (\ref{eq:lambda-bound}). Using the (still very crude) estimate $e^{-\Omega(n)} << 10^{-\ell-1}$, (\ref{eq:lambda-bound}) implies
\begin{eqnarray*}
\E[1 - (1-\phi(X_{j+1}))^2] & \geq & \E[(1-\epsilon) \phi(X_j-n m^{-1/2}) - e^{-\Omega(n)}] - \frac{\E[P_j] - p^*}{\lambda} \\
& \geq & \E[(1-\epsilon) \phi(X_j-10^{-\ell}) - 10^{-\ell-1}] - 10^{-\ell-1} \\
& \geq & \E[(1-\epsilon) \phi(X_j-10^{-\ell}) - 10^{-\ell}].
\end{eqnarray*}
\end{proof}

\subsection{The convex separation argument}

Next, we use a geometric argument, essentially Farkas' lemma in 2 dimensions, which shows that the bounds for varying multipliers $\lambda', \lambda'' \geq 0$ allow us to obtain separate bounds on value and price.

\begin{lemma}
\label{lem:convex}
Consider a $(1-\epsilon)$-approximately truthful-in-expectation mechanism for combinatorial auctions with $n=2^\ell$ players and $m=400^\ell$ items. Let $\phi:[0,1] \rightarrow [0,1]$ be a non-decreasing concave function. If $(X_j,P_j)$ has a distribution in the closure of the level-$j$ menu $\overline{\cM_j}$, then there is $(X_{j+1},P_{j+1})$ with a distribution in the closure of the level-$(j+1)$ menu $\overline{\cM_{j+1}}$ such that
$$\E[P_{j+1}] \leq \E[P_j] $$ and
$$\E[1 - (1-\phi(X_{j+1}))^2] \geq \E[(1-\epsilon) \phi(X_j-10^{-\ell}) - 10^{-\ell}].$$
\end{lemma}

\begin{figure}[!h]
\label{fig:covex-sep}
\centering
\begin{tikzpicture}[scale=.60,pre/.style={<-,shorten <=2pt,>=stealth,thick}, post/.style={->,shorten >=1pt,>=stealth,thick}]

\shadedraw [black,ball color=gray] (-5,-2) .. controls (-7,2) and (-2,5) .. (4,4) .. controls (2,1) and (1,1) .. (-5,-2);
\shadedraw [white,shading=axis,shading angle=45] (0,0) rectangle (10,-10);

\fill (0,0) circle (4pt);
\draw (-10,0) -- (10,0);
\draw (0,-10) -- (0,6);

\draw [dashed] (-10,-5) -- (10,5);

\draw (1.2,-0.5) node {$(q_j,p_j)$};

\draw (-4.5,3) node {${\cal Q}_{j+1}$};

\end{tikzpicture}
\caption{The convex separation argument.}
\label{fig:convex-sep}
\end{figure}
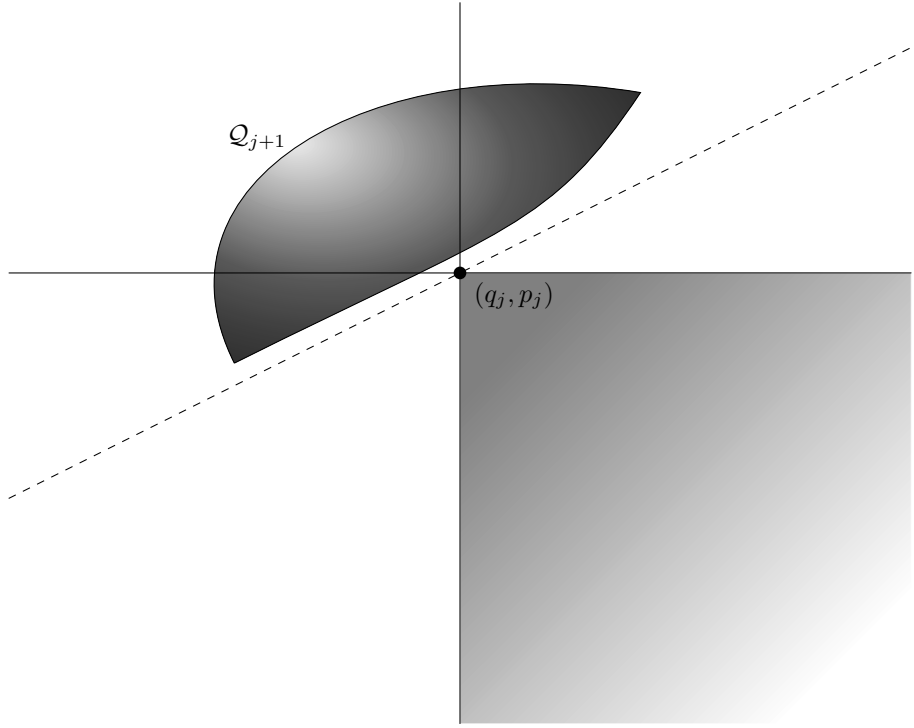

\begin{proof}
Denote $q_j = \E[(1-\epsilon) \phi(X_j-10^{-\ell}) - 10^{-\ell}]$ and $q_{j+1} = \E[1 - (1-\phi(X_{j+1}))^2]$. Set also $p_{j} = \E[P_{j}]$ and $p_{j+1} = \E[P_{j+1}]$. 
By Lemma~\ref{lem:lambda-gap}, for any $(X_{j},P_{j}) \in \overline{\cM_j}$ and any $\lambda',\lambda'' \geq 0$, there is $(X_{j+1},P_{j+1}) \in \overline{\cM_{j+1}}$  such that $\lambda' q_{j+1} - \lambda'' p_{j+1} \geq \lambda' q_j - \lambda'' p_j$.
Using these transformations, let us map the set $\overline{\cM_j}$ to 
$$ \cQ_j = \{ (q_j, p_j): (X_j, P_j) \in \overline{\cM_j} \}.$$
and map $\overline{\cM_{j+1}}$ to
$$ \cQ_{j+1} = \{ (q_{j+1}, p_{j+1}): (X_{j+1}, P_{j+1}) \in \overline{\cM_{j+1}} \}. $$
Both $\cQ_j$ and $\cQ_{j+1}$ are closed convex sets, because they are the images of the closed convex sets $\overline{\cM_j}, \overline{\cM_{j+1}}$ under a linear map (the map being the expectation of a certain function over a distribution; this is linear
as a function of the distribution even though the function is non-linear).

By this transformation, we have reduced the proof to a geometric question in the plane (see Figure~\ref{fig:convex-sep}): Given $(q_j,p_j)$, assume that for any $\lambda', \lambda'' \geq 0$, there is $(q_{j+1}, p_{j+1}) \in \cQ_{j+1}$ such that $\lambda' q_{j+1} - \lambda'' p_{j+1} \geq \lambda' q_j - \lambda'' p_j$.
Is it possible that there is no point $(q_{j+1}, p_{j+1}) \in \cQ_{j+1}$ such that $q_{j+1} \geq q_j$ and $p_{j+1} \leq p_j$?

Suppose that there is no such point in $\cQ_{j+1}$.
This means that $\cQ_{j+1}$ and $\{(q,p): q \geq q_j, p \leq p_j\}$ are disjoint.
Since these are closed convex sets, they can be separated by a line.
This line cannot have a negative slope, otherwise it would intersect the quadrant $\{(q,p): q \geq q_j, p \leq p_j\}$. Such a separating line gives $\lambda', \lambda'' \geq 0$ such that $\lambda' q_{j+1} - \lambda'' p_{j+1} < \lambda' q_j - \lambda'' p_j$ for all $(q_{j+1}, p_{j+1}) \in \cQ_{j+1}$. However, this contradicts the assumption above. Hence there is a point $(q_{j+1}, p_{j+1}) \in \cQ_{j+1}$ such that $q_{j+1} \geq q_j$ and
 $p_{j+1} \leq p_j$.
\end{proof}

\subsection{Putting it all together}

In this section, we finish the proof of our main hardness result for combinatorial auctions.

\begin{proof}[Proof of Theorem~\ref{thm:CA-hardness}]
Let $\epsilon>0$ and $\delta>0$ be the constants provided by Lemma~\ref{lem:level-bound}. Let the number of players be $n=2^\ell$ and the number of items $m = 400^\ell$. Suppose there is a $(1-\epsilon)$-approximately truthful-in-expectation mechanism that provides a $c$-approximation in social welfare, where $c = 1/n^\gamma = 2^{-\gamma \ell}$ for some constant $\gamma>0$.

Consider the basic instance (Section~\ref{sec:basic-inst}). Choose a special player and fix the remaining valuations, based on Lemma~\ref{lem:basic-inst}. Let $R^{(0)}$ be the random set allocated to the special player, $P_0$ the respective price, $A^{(0)}$ his desired set, $X_0 = \frac{|R^{(0)} \cap A^{(0)}|}{|A^{(0)}|}$, $c_0 = \E[X_0]$ and $p_0 = \E[P_0]$. 
Lemma~\ref{lem:basic-inst} implies $c_0 \geq c/4 - \omega$. We set $\omega = c/8$. Then
 $c_0 \geq c/8 = 1/(8 n^\gamma) = 2^{-\gamma \ell - 3}$.

Now consider the distribution menus at different levels and their closures $\overline{\cM_j}$ (Section~\ref{sec:prep}).
Let us define $\cX_j$ to be the collection of random variables $X_j$ such that $(X_j,P_j)$ is in $\overline{\cM_j}$ for some price $P_j$ such that $\E[P_j] \leq p_0$. As discussed above, we have $X_0$ in $\cX_0$ such that $\E[X_0] = c_0 \geq 2^{-\gamma \ell - 3} \geq 2^{-\ell}$ for $\ell$ sufficiently large. Also, Lemma~\ref{lem:convex} says that for any $X_j$ in $\cX_j$ and any non-decreasing concave $\phi:[0,1] \rightarrow [0,1]$, we have $X_{j+1}$ in $\cX_{j+1}$ such that
$$\E[1 - (1-\phi(X_{j+1}))^2] \geq \E[(1-\epsilon) \phi(X_j-10^{-\ell}) - 10^{-\ell}].$$
In other words, the collections $\cX_0, \cX_1, \ldots, \cX_\ell$ satisfy the assumptions of Lemma~\ref{lem:level-bound}.
Hence, by Lemma~\ref{lem:level-bound} for $j=\ell$, there is $X_\ell$ in $\cX_\ell$ and $\alpha_\ell \in [0,1]$ such that
$$ \alpha_\ell (\E[\phi_{\alpha_\ell}(X_\ell)])^{1+\delta} \geq \left( \frac{1+\delta^2}{2} \right)^\ell c_0^{1+\delta}.$$
Recall that $\phi_{\alpha_\ell}(t) = \min \{ \frac{t}{\alpha_\ell}, 1\}$. Therefore, we have 
$$  \E[X_\ell] \geq \alpha_\ell \E[\phi_{\alpha_\ell}(X_\ell)] \geq  \alpha_\ell (\E[\phi_{\alpha_\ell}(X_\ell)])^{1+\delta} \geq \left( \frac{1+\delta^2}{2} \right)^\ell c_0^{1+\delta} \geq 2^{\delta^2 \ell - \ell} c_0^{1+\delta}.$$
Since $X_\ell$ is in $\cX_\ell$, the respective price is bounded by $\E[P_\ell] \leq p_0$. 
Now consider an expression related to the utility the special player would derive in the basic instance:
$$ \E[(1-\epsilon) \omega m X_\ell - P_\ell] \geq
 (1-\epsilon) \omega m 2^{\delta^2 \ell - \ell} c_0^{1+\delta} - p_0.$$
The distribution of $(X_{\ell}, P_\ell)$ might not be on the actual menu $\cM_\ell$ of the special player; however, since it is in the closure of its convex hull, there exists a pair $(\tilde{X}_{\ell}, \tilde{P}_\ell)$ with a distribution in $\cM_\ell$ such that
$$ \E[(1-\epsilon) \omega m \tilde{X}_\ell - \tilde{P}_\ell]
 > (1-2\epsilon) \omega m 2^{\delta^2 \ell - \ell} c_0^{1+\delta} - p_0.$$
(If not, we get a contradiction since if the reverse inequality holds for $\cM_\ell$, it holds also for the closure $\overline{M}_\ell$.) 
The random variable $\tilde{X}_\ell$ represents a random set $\tilde{R}^{(\ell)}$, possibly allocated to the special player at price $\tilde{P}_\ell$: $\tilde{X}_\ell = \frac{|\tilde{R}^{(\ell)}|}{|A^{(\ell)}|} = \frac{1}{m} |\tilde{R}^{(\ell)}|$. 

Now let us go back to the basic instance. Considering that the valuation of the special player in the basic instance satisfies $v^*_i(S) \geq \omega|S|$, we obtain
$$ \E[(1-\epsilon) v^*_i(\tilde{R}^{(\ell)}) - \tilde{P}_\ell]
 \geq \E[(1-\epsilon) \omega m \tilde{X}_{\ell} - \tilde{P}_\ell]
 > (1-2\epsilon) \omega m 2^{\delta^2 \ell - \ell} c_0^{1+\delta} - p_0.$$ 
Using $c_0 \geq c/8 = \omega = 2^{-\gamma \ell - 3}$, we get
\begin{equation}
\label{eq:menu-item}
\E[(1-\epsilon) v^*_i(\tilde{R}^{(\ell)}) - \tilde{P}_\ell]
> (1-2\epsilon) \left( \frac{c}{8} \right)^{1+\delta} 2^{\delta^2 \ell - \ell} m c_0 - p_0
= (1-2\epsilon) 2^{\delta^2 \ell - (1+\delta) (\gamma \ell + 3) - \ell} m c_0 - p_0.
\end{equation}
On the other hand, the set $R^{(0)}$ actually allocated under declared valuation $v^*_i$ gives
$$ \E[v^*_i(R^{(0)})] = \E[|R^{(0)} \cap A^{(0)}|] + \omega \E[|R^{(0)} \setminus A^{(0)}|]
 \leq \frac{m}{n} \E[X_0] + \omega \E[|R^{(0)}|] \leq \frac{2m}{n} \E[X_0] $$ 
using again Lemma~\ref{lem:basic-inst} to say that $\frac{m}{n} \E[X_0] = \E[|R^{(0)} \cap A^{(0)}|] \geq (c/4-\omega) \E[|R^{(0)}|] = \omega \E[|R^{(0)}|] $. Therefore, since $\E[X_0] = c_0$ and $\E[P_0] = p_0$,
\begin{equation}
\label{eq:true-item}
\E[v^*_i(R^{(0)}) - P_0] \leq \frac{2m}{n} \E[X_0] - \E[P_0] = 2^{1-\ell} m c_0 - p_0.
\end{equation}
Since $\tilde{R}^{(\ell)}$ is a random set the special player could receive at price $\tilde{P}_\ell$ if he had declared a suitable valuation, $(1-\epsilon)$-approximate truthfulness implies that
$$ \E[v^*_i(R^{(0)}) - P_0] \geq \E[(1-\epsilon) v^*_i(\tilde{R}^{(\ell)}) - \tilde{P}_\ell]. $$
Considering (\ref{eq:menu-item}) and (\ref{eq:true-item}), this implies
$$ 2^{1-\ell} > (1-2\epsilon) 2^{\delta^2 \ell - (1+\delta)(\gamma \ell+3)-\ell}.$$
We conclude that $\gamma \geq \frac{\delta^2}{1+\delta}$, otherwise we get a contradiction for a large enough $\ell$.
\end{proof}

\section{Chernoff bound for bisections}

\begin{lemma}
\label{lem:bisect-chernoff}
Suppose $S$ is a fixed subset of $[m']$, and $(A,B)$ a random partition of $[m']$, chosen uniformly among all partitions where $|A| = |B| = m'/2$. Then
$$ \Pr\left[ ||S \cap A| - |S \cap B|| > \beta m'] \right] < 4 e^{-\beta^2 m' / 2}.$$
\end{lemma}

\begin{proof}
We use the fact that $A$ has distribution very close to a uniformly random subset of $[m']$ (where elements appear independently with probability $1/2$). More precisely, we couple the two distributions as follows. Let $A$ be a random set of size $m'/2$, $B$ its complement, and let $X$ be a binomial random variable $Bi(m',1/2)$. Let $R$ be a random set chosen as follows: if $X \leq m'/2$, take a random subset of $A$ of size $X$. If $X > m'/2$, take the union of $A$ and $X-m'/2$ random elements from $B$. This defines a set $R$ which is uniformly random. Hence, by the Chernoff bound (see e.g. \cite[Theorem A.1.16]{AlonSpencer}),
$$ \Pr[|R \Delta A| > \alpha m'] = \Pr[Bi(m',1/2) \notin [m'/2-\alpha m', m'/2+\alpha m']] < 2 e^{-2 \alpha^2 m'}.$$
Similarly, $S \Delta R$ has the distribution of a uniformly random set (because $S$ is fixed), and hence
$$ \Pr[|S \Delta R| \notin [m'/2 - \alpha m', m'/2 + \alpha m']] < 2 e^{-2 \alpha^2 m'}.$$
Using the triangle inequality $|S \Delta A| \leq |S \Delta R| + |R \Delta A|$, we get
$$ \Pr[|S \Delta A| \notin [m'/2 - 2 \alpha m', m'/2 + 2 \alpha m']] < 4 e^{-2 \alpha^2 m'}.$$
The lemma follows by taking $\alpha = \beta/2$, since $|S \cap A| - |S \cap B| = |A| - |S \Delta A| = \frac{m'}{2} - |S \Delta A|.$
\end{proof}

\section{Product composition of submodular functions}

\begin{lemma}
\label{lem:submod-product}
Let $f_1,f_2:2^M \rightarrow [0,1]$ be monotone submodular. Then
$$ f(S) = 1 - (1-f_1(S)) (1-f_2(S)) $$
is also monotone submodular.
\end{lemma}

\begin{proof}
Let $g_1(S) = 1-f_1(S)$, $g_2(S) = 1-f_2(S)$; these are non-negative non-increasing supermodular functions. Clearly, $g(S) = g_1(S) g_2(S)$ is also non-increasing. Our goal is to prove that $g(S) = g_1(S) g_2(S)$ is supermodular, which implies the claim. By the properties of $g_1, g_2$, we get for any $i,j \notin S$
$$ g_1(S) (g_2(S) - g_2(S+i)) \geq g_1(S) (g_2(S+j) - g_2(S+i+j)) \geq g_1(S+j) (g_2(S+j) - g_2(S+i+j)) $$
and
$$ (g_1(S) - g_1(S+i)) g_2(S+i) \geq (g_1(S+j) - g_1(S+i+j)) g_2(S+i) \geq (g_1(S+j) - g_1(S+i+j)) g_2(S+i+j).$$
Adding up these two inequalities, we get the condition of supermodularity for $g(S) = g_1(S) g_2(S)$:
$$ g_1(S) g_2(S) - g_1(S+i) g_2(S+i) \geq g_1(S+j) g_2(S+j) - g_1(S+i+j) g_2(S+i+j).$$
\end{proof}


\end{document}